\documentclass[a4paper,11pt]{article}

\usepackage[margin=1in]{geometry}
\usepackage{amsthm}
\usepackage{amssymb}
\usepackage{amsmath}
\usepackage{mathrsfs}
\DeclareMathAlphabet{\mathpzc}{OT1}{pzc}{m}{it}
\usepackage{tikz-cd} 
\usepackage{enumitem}
\usetikzlibrary{graphs,decorations.pathmorphing,decorations.markings}

\usepackage{tensor}
\newcounter{dummy} \numberwithin{dummy}{subsection}
\newtheorem{myth}[dummy]{Theorem}
\newtheorem*{mylem*}{Lemma}

\newtheorem{mydef}[dummy]{Definition}
\newtheorem{myprop}[dummy]{Proposition}

\newtheorem{mylem}[dummy]{Lemma}

\theoremstyle{definition}
\newtheorem{myeg}[dummy]{Example}
\DeclareMathOperator{\id}{id}
\newcommand{\R}{\mathbb{R}}

\DeclareMathOperator{\dom}{dom}
\DeclareMathOperator{\arctanh}{arctanh}
\usepackage{hyperref}
\hypersetup{
    colorlinks=true,
    linkcolor=blue,
    filecolor=blue,      
    urlcolor=blue,
    citecolor = blue,
}
 
\newsavebox{\pullback}
\sbox\pullback{%
\begin{tikzpicture}%
\draw (0,0) -- (2ex,0ex);%
\draw (2ex,0ex) -- (2ex,2ex);%
\end{tikzpicture}}

\newcounter{examp}

\newcounter{exampa}

\newcounter{exampb}

\def\mathen{{\hbox{-}}}
\usepackage{stmaryrd}
\begin{document} 
{\center{\Huge Inverse Higgs phenomena as duals of holonomic constraints}\\~\\
{\Large Ben Gripaios, and Joseph Tooby-Smith}\\~\\
Cavendish Laboratory, University of Cambridge, J.~J.~Thomson Ave, Cambridge, UK\\ ~\\
Emails: gripaios@hep.phy.cam.ac.uk and jss85@cam.ac.uk}
~\\
~\\
~\\
\begin{abstract}
The inverse Higgs phenomenon, which plays an important r\^ole in physical systems with Goldstone bosons (such as the phonons in a crystal) involves nonholonomic mechanical constraints. By formulating field theories with symmetries and constraints in a general way using the language of differential geometry, we show that many examples of constraints in inverse Higgs phenomena fall into a special class, which we call coholonomic constraints, that are dual (in the sense of category theory) to holonomic constraints. Just as for holonomic constraints, systems with coholonomic constraints are equivalent to unconstrained systems (whose degrees of freedom are known as essential Goldstone bosons), making it easier to study their consistency and dynamics.  The remaining examples of inverse Higgs phenomena in the literature require the dual of a slight generalisation of a holonomic constraint, which we call (co)meronomic. Our formalism simplifies and clarifies the many {\em ad hoc} assumptions and constructions present in the literature. In particular, it identifies which are necessary and which are merely convenient. It also opens the way to studying much more general dynamical examples, including systems which have no well-defined notion of a target space. 
\end{abstract}
\section{Introduction}
This work describes constraints in field theories with symmetry, in a general way, using the language of differential geometry. Of particular interest is the special case in which the symmetry group acts transitively on the space carrying the fields. This includes theories of Goldstone bosons exhibiting the so-called `inverse Higgs phenomenon' (a name which can surely be bettered), in which the presence of constraints involving derivatives of the fields implies that Goldstone's theorem no longer holds, leading to richer possibilities for dynamics~\cite{Ivanov_Ogievetsky_1975}. Such constraints are generic, due to the simple fact that no symmetry can act transitively on the fields and their derivatives, once we include enough derivatives.  A well-known example are the phonons occurring in crystalline media.

Our main motivation for the work is not the pursuit of generality for its own sake, but rather to show that many of the apparently {\em ad hoc} constructions existing in the literature on the inverse Higgs phenomenon are, in fact, very natural, when viewed with a sufficient level of abstraction. Doing so also makes it easier to see which of the various assumptions made are necessary for physical consistency and which are merely convenient. 

Perhaps the most important insight we obtain is the following. In the special case where the symmetry acts transitively, any constraint is necessarily nonholonomic. Such constraints are notoriously difficult to deal with in general, even in classical mechanics (an infamous example being the motion of a bicycle).  
By suitably reformulating the more familiar notion of a holonomic constraint in our framework, we will see that there exists a special class of nonholonomic constraints that are dual (in the sense of category theory) to holonomic constraints, which we thus call {\em co}holonomic constraints. A glance at the precise definitions in \ref{def:holonomic} and \ref{def:coholonomic} shows that the duality is somewhat fiddly at the level of the aforementioned `space carrying the fields' (which is, mathematically, a fibred manifold), but it reduces to the following simple statement at the level of the kinematic degrees of freedom of the physical theory: a system with a holonomic constraint is equivalent to an unconstrained system defined on a subobject, while a system with a coholonomic constraint is equivalent to an unconstrained system defined on a quotient object.
The first part of the statement (which is, mathematically, a theorem about sheaves) corresponds, at an elementary level, to the notion of `solving the constraint to eliminate redundant degrees of freedom', while its dual corresponds to the familiar notion that one can consider just `essential Goldstone bosons'.
Because theories constrained in such ways are kinematically equivalent to unconstrained ones, no new issues of physical consistency arise and no new difficulties are encountered in formulating and studying dynamics (unlike for bicycle motion).

Remarkably, it turns out that every example of the inverse Higgs phenomenon that we have been able to find in the literature involves the dual of either a holonomic constraint or, in just a few cases, of a slight generalisation thereof, which we call (co)meronomic constraints (definitions are given in \ref{def:meronomic} and \ref{def:comeronomic}). Systems with (co)meronomic constraints are not obviously equivalent to unconstrained systems and so we must worry about issues of physical consistency. Here, we content ourselves with establishing just two basic consistency properties enjoyed by such constrained systems, namely that they satisfy basic locality requirements and that  
local degrees of freedom exist at every spacetime point (in the language of sheaf theory, we require that the degrees of freedom form a sheaf whose stalks are not empty). 

To describe the other features of our approach, it is perhaps easiest to sketch the basic ingredients. We begin, in \S\ref{sec:Constraints}, by describing constraints in field theories without regard to symmetry. Rather than using local coordinates, as in the physics literature, we use a coordinate free approach, which not only allows us to take global considerations into account, but also clarifies exactly which mathematical structures are being made use of. 

In the most basic examples of field theories, the fields are smooth maps from some `spacetime' manifold to some `target' manifold, so the `space carrying the fields' can be taken to be simply the product of the two manifolds. We generalise by replacing this product by a fibred manifold. Precise definitions will follow, but for now it is enough to note that fibred manifolds are the most general objects that (locally) admit smooth sections, which can serve as the local degrees of freedom ({\em i.e.} the `fields' of the field theory). Fibred manifolds generalise the more familiar notion of fibre bundles, in that over each point in spacetime there is a well-defined fibre. But unlike fibre bundles, the fibres over different points in spacetime may not even have the same homotopy type, let alone diffeomorphism class, so there is no meaningful notion, even locally, of a `target space'. 

Fibred manifolds form a category and we will see that many of the constructions required for dynamics are conveniently understood using the language of category theory. For example, there is a functor -- the $r$th-jet functor -- which sends a fibred manifold to its $r$th-jet manifold, which is itself a fibred manifold encoding the notion of the derivatives of sections of order up to $r$, in a coordinate-free way. Consistent dynamical constraints may be described as certain subobjects of the jet manifold and we show how holonomic and meronomic constraints (and their duals) can be understood in this way. Consistency, for us, amounts to insisting that the sections that are compatible with the constraint form a sheaf (such that locality is obeyed) whose stalks are non-empty (meaning that local degrees of freedom exist at every point in spacetime). 

In \S\ref{sec:homogeneous}, we introduce the notion of symmetry, via a Lie group action on the fibred manifold. A great deal of simplification arises in the special case where the action is transitive and equivariant with respect to the projection onto spacetime, which we call a fibrewise action (an example is the galilean symmetry of a non-relativistic particle). In such a case, both the fibred manifold and its jet manifolds take the form of fibre bundles associated to the $L$-principal bundle $G \to G/L$, for Lie groups $L \subset G$. The category of such bundles (called homogeneous bundles in the mathematical literature) is equivalent to the category of manifolds equipped with an action of the group $L$. This simple statement extends and makes rigorous
 the physicist's vague notion (put forward in \cite{Coleman:1969sm,Callan:1969sn}) that `in studying sigma models based on a target space $G/L$, $L$ invariance implies $G$ invariance'. It also shows that some constructions used in the literature on sigma models, such as connections and vielbeins, are unnecessary. We describe a number of examples with group actions of this type.

More generally, is an unavoidable fact that starting from a group action on a fibred manifold, in general only a partial group action 
is induced on its jet manifolds (the Poincar\'{e} symmetry of a relativistic particle is an example). It therefore makes sense to work with partial group actions from the off in the general case, which we do in \S\ref{sec:partial_actions}. Though the resulting mathematics is technically rather cumbersome, the results are conceptually straightforward, thanks to the category-theoretic nature of our earlier constructions. We also discuss a number of examples with partial group actions.

In order to ease the burden on the reader, the more technical proofs have been deferred to the Appendices.

Our discussion is purely at the level of kinematics; in particular, we do not discuss how dynamics can be specified in the form of an action (in the physics sense of the word). In all examples we study, this is, however, straightforward: the action is determined by choosing a differential form on the constraint manifold (which is a submanifold of the $r$th jet manifold) whose degree coincides with the dimension of the spacetime manifold. The action is then evaluated on a section ({\em i.e.} a field) by pulling back the differential form along the section and integrating over spacetime. One complication is that many such forms yield actions that are trivial in the sense that they are either identically zero or do not contribute to the equations of motion. In the presence of symmetry, this makes the classification of invariant dynamical theories tricky, because the set of such theories includes those whose action is not invariant under the group transformations, but rather shifts by such a trivial action.
\section{Mathematical prerequisites} \label{sec:mathematical_prerequisites}
\subsection{Motivating ideas} \label{subsec:Motivating_Ideas}
In this Section, we describe the required mathematical machinery. Since this goes somewhat beyond the usual physicist's curriculum, we begin by describing in an informal way what it is, and why it is needed.

Since physics is based upon local measurements in spacetime, it is natural to work using explicit local coordinates $x^\mu$ in spacetime. But since the specific choice of such coordinates is made at the observer's whim, the physics itself should not depend upon the choice. Coupled with the desire to be able to describe spacetimes that are not contractible,
we are naturally led to the concept of a spacetime manifold $X$, which should moreover have a smooth structure so that we can define a dynamical action involving derivatives. (In what follows, almost everything will be taken to be smooth, so we omit reference to it unless there is a risk of confusion.)

A manifold comes naturally equipped with open sets and it is perhaps helpful to visualise
these as `laboratories without walls' in which observers can carry out their local measurements. The `without walls' condition, or more precisely the condition that a set be open, ensures that observers whose laboratories intersect can compare measurements without having to worry about annoyances such as boundary conditions, {\em \&c}.

Now that we have our mathematical model of spacetime, we may consider the degrees of freedom, or fields, of a field theory living on it. In the approach using explicit local coordinates, these take the form of maps $x^\mu \mapsto y^a(x^\mu)$, but there are several reasons why, in the approach using manifolds, we should not simply replace this by a map from $X$ to some other manifold representing an internal or `target' space. One is that there are known examples in physics, namely gauge theories, where this is not the case (there, the matter fields are instead sections of a fibre bundle). A second reason is that this construction amounts to the assertion that the internal spaces at each spacetime point can be canonically identified with one another, which seems inconsistent with the general expectation that physics should not feature `action at a distance'. A third reason is that this structure is anyway not preserved once we take derivatives into account, as we shall see below.

We instead take the fields of a field theory (at least in the unconstrained case) to be local sections of a fibred manifold. A {\em fibred manifold} consists of a pair of manifolds, $X$ -- the {\em base} -- and $Y$ -- the {\em total space} -- together with a surjective submersion $\pi:Y \to X$ and a local section is a smooth map $\alpha:U \to Y$ on some open subset $U \subseteq X$ which is a right inverse to $\pi$.\footnote{Suitable references are \cite{lee_2009,Kolar_1993,Saunders_1989,Sardanashvily:2009br,olver2000applications}.}

A fibred manifold is perhaps best viewed as a generalisation of the more familiar notion of a fibre bundle. Indeed, just as for a fibre bundle, the inverse image $\pi^{-1}(x)$ of a point $x \in X$ in the base is itself a manifold, which we call the {\em fibre at $x$}. But unlike a fibre bundle, the fibres over different points need not have the same homotopy type, let alone diffeomorphism class.\footnote{An example is given by $Y=\mathbb{R}^2-\{0\}$, and $X=\mathbb{R}$, with the projection onto the first factor. The fibre at $x=0$ does not have the same homotopy type as elsewhere.} Since we interpret the fibre in physics as the internal space over the spacetime point $x$, we see that fibred manifolds allow for dramatically different field theories than those we are used to.

Nevertheless, such theories are compatible with the usual consistency requirements that we impose on physical theories. Indeed, just as for a fibre bundle, the fact that $\alpha$ is a right inverse to $\pi$ guarantees that the sections collectively form a sheaf on $X$ and so satisfy basic locality requirements. Most of these conditions ({\em i.e.} those for a presheaf) seem almost too obvious to mention; for example, we require that sections ({\em i.e.} fields) defined on an open set ({\em i.e.} in a laboratory)  restrict to fields defined on an open subset ({\em i.e.} in a smaller laboratory contained in the original one). But one -- the sheaf condition -- is not so trivial: it requires that given sections agreeing on the intersection of some collection of open sets, there exists a unique section on the union of that collection. It is thus a necessary precondition on kinematics for different observers to be able to compare measurements.

Moreover, just as for fibre bundles, the fact that $\pi:Y \to X$ is a surjective submersion guarantees that a local section exists in some neighbourhood of every point of $X$. Because of the presheaf condition, local sections will then exist on all subneigbourhoods and we interpret this as capturing the physically-reasonable requirement that local degrees of freedom should exist in a sufficiently small neighbourhood of each spacetime point.

In fact, a stronger statement is possible: a fibred manifold admits a local section not just at every point in $X$, but through every point in $Y$. Indeed, it is possible to choose {\em adapted coordinates}
$(x^\mu,y^\alpha)$ in a neighbourhood of every point of $Y$ such that $\pi$ restricts to $(x^\mu,y^\alpha) \mapsto x^\mu$, whose sections are equivalent to functions $x^\mu \mapsto y^a(x^\mu)$. This brings us back to our starting point, showing that fibred manifolds give us a global, coordinate-free notion of (unconstrained) fields that is compatible with locality.

The introduction of constraints will require us to reexamine this picture. Indeed, a constraint will restrict us to a subset of the local sections, namely those that satisfy the constraint. We will need to check that our basic physical requirements are still satisfied and this will require heavy use of the theory of sheaves. In particular, we need to ensure that locality is preserved, {\em i.e.} that the sections still form a sheaf, since the existence part of the sheaf condition is no longer obviously satisfied. Moreover, it is also obviously the case that sections will no longer exist through every point of $Y$ (consider the case of a holonomic constraint, part of the data of which is a submanifold of $Y$) and so we will need to ensure that local sections exist at least at every point of $X$, as we required before. This is equivalent to the requirement that the stalks of the sheaf are non-empty.

Mostly, we will not actually work with sheaves, but rather with the equivalent notion of 
\'etal\'e spaces, since they simplify the discussion of stalks as well as group actions. 
An \'etal\'e space can be given a physical motivation as follows.
Imagine an observer at $x\in X$, whose laboratory is arbitrarily small. Such an observer will not be able to distinguish local sections $\alpha:U\rightarrow Y$, and $\beta:V\rightarrow Y$, for $U,V \ni x$, if there is an open subset $W$ with $x \in W\subseteq U\cap V$, such that $\alpha\circ \iota_{W,U}=\beta\circ \iota_{W,V}$, where $\iota_{W,U}$ and $\iota_{W,V}$, are the inclusion maps. Thus, the observer is sensitive only to the equivalence class $[\alpha]_x$ of local sections, where $[\alpha]_x=[\beta]_x$ if $\alpha$ and $\beta$ agree in the way just described. An equivalence class at $x$ is called a {\em germ at $x$} and the set of such germs is called the  {\em stalk at $x$}.
The {\em \'etal\'e space} $(\Gamma Y,\Gamma \pi)$ is then defined as follows. The topological space $\Gamma Y$ is, as a set, the disjoint union over $x \in X$ of the stalks, equipped with the unique topology making the map $\Gamma\pi:\Gamma Y\rightarrow X:[\alpha]_x\mapsto x$ into a local homeomorphism.\footnote{
in this topology, given $U \in X$ and a local section $\alpha:U\rightarrow Y$, the set $\{[\alpha]_x|x \in U\} \subset \Gamma Y$ is open and the set of such open sets obtained by varying $U$ and $\alpha$ forms a basis for the topology.} 
In physics terms, the \'etal\'e space encodes the totality of information available to observers with arbitrarily small laboratories.

Evidently, the germs making up the points of $\Gamma Y$ remember all the derivatives (in some adapted coordinates) of local sections so contain at least enough information to allow us to define constraints involving any finite number of derivatives (as well as an action to any finite order in some effective field theory expansion).\footnote{In fact they contain more information, as the following example shows: let $Y=\mathbb{R}^2$ and $X=\mathbb{R}$, with the standard projection. Then $\alpha(x)=(x,e^{-1/x^2})$  for $x\ne 0$ and  $\alpha(0)=0$, has the same Taylor expansion as $\beta(x)=0$ at $x=0$ but $[\alpha]_0\ne[\beta]_0$.} But the topological space $\Gamma Y$ is not even Hausdorff in general, so cannot be given a smooth structure. To apply the full power of differential geometry to the discussion of constraints, we need to recover such a structure. This can be done by defining coarser equivalence classes, denoted $j^r_x\alpha$, with $j^r_x\alpha=j^r_x\beta$ if and only if the derivatives of $\alpha$ and $\beta$ (computed in some adapted coordinates, the choice of which does not affect the result) agree up to and including the $r$th order. The set of all equivalence classes $j^r_x\alpha$ for all $x\in X$ is denoted $J^rY$. The set $J^rY$ can be given a smooth structure making it into a manifold, called the {\em $r$th-jet manifold}, and making the map $\pi^r:J^rY\rightarrow X:j^r_x\alpha \mapsto x$ a surjective submersion (an observation which is vital for our discussion). If $(x^\mu,y^a)$ are adapted coordinates on $Y$, and locally $\alpha:x^\mu\mapsto (x^\mu,y^a(x^\mu))$, $J^rY$ admits {\em induced coordinates}, which for $J^1Y$ take the form $(x^\mu,y^a,y^a_\mu)$ such that $j^r_x\alpha$ corresponds to the point $(x^\mu,y^a(x^\mu),\partial_\mu y^a(x^\mu))$, with an obvious generalisation to $J^{r>1}Y$. It is these induced coordinates that physicists use to write down lagrangians, but the approach using jet bundles has the advantage of being coordinate free. We remark that, even if one starts from a fibred manifold in the form of a product $Y = F \times X$, the jet manifold need not take the form of a product $J^rY = F^\prime \times X$. This shows, as we vaguely alluded to earlier, that even for physical theories whose degrees of freedom are maps from spacetime to a target, one must pass to the more general fibred manifold picture once derivatives are included.

\subsection{Categorical preliminaries}\label{subsec:Categoric_Prelims}
Many of our constructions are conveniently described using the language of category theory, whose rudiments we now describe.

 A {\em category} ${\tt C}$ is a collection of objects and morphisms between those objects satisfying a series of axioms. Namely, for each object $C$ there is an identity morphism $\id_{C}:C\rightarrow C$ and we can compose any morphism from $C$ with any morphism to $C$, subject to the rules that composition is associative and that pre- or post-composing a morphism with the identity morphism returns the original morphism.
Examples are the category {\tt Set}, whose objects are sets and whose morphisms are functions, the category {\tt Top}, whose objects are topological spaces and whose morphisms are continuous maps, and the category  {\tt Man}, whose objects are smooth manifolds and whose morphisms are smooth maps.

Given a pair of categories ${\tt C, C^\prime}$, a {\em functor} $F: {\tt C}\rightarrow {\tt C^\prime}$ is a mapping of each object $C$ in $\tt C$ to an object $F(C)$ in $\tt C^\prime$ and a mapping of each morphism $f:C \rightarrow \tilde{C}$ in $\tt C$ to a morphism $F(f):F(C)\rightarrow F(\tilde{C})$ in $\tt C^\prime$ that preserves identities and composition.
We have, for instance, functors ${\tt Man}\rightarrow {\tt Top}$ and ${\tt Top}\rightarrow {\tt Set}$ that simply forget the extra structure. 

Given a category {\tt C}, its {\em opposite category} ${\tt C}^{\mathrm{op}}$ has the same objects as {\tt C}, but all morphisms have their sources and targets swapped.
A functor from ${\tt C}^{\mathrm{op}}$ to $\tt C^\prime$ is often called a contravariant functor from {\tt C} to $\tt C^\prime$. 

Given a pair of functors $F,F^\prime: {\tt C} \to {\tt C^\prime}$, a {\em natural transformation} $\eta: F \Rightarrow F^\prime$ is, for every object $C$ in $\tt C$ a morphism $\eta_C:F(C)\rightarrow F^\prime(C)$, such that, for every morphism $f:C\rightarrow \tilde{C}$, the diagram
\begin{equation}
\begin{tikzcd}
F(C)\arrow{r}{\eta_C}\arrow[swap]{d}{F(f)} &F^\prime(C)\arrow{d}{F^\prime(f)}\\
F(\tilde C)\arrow[swap]{r}{\eta_{\tilde C}} & F^\prime(\tilde C)
\end{tikzcd}
\end{equation}
commutes. A {\em natural isomorphism} is a natural transformation for which each morphism $\eta_C$ is an isomorphism in $\tt C^\prime$. 

An {\em equivalence of categories} $\tt C, C^\prime$ is a pair of functors $F:{\tt C}\rightarrow {\tt C}^\prime$ and $F^\prime:{\tt C}^\prime \rightarrow {\tt C}$, such that there exist natural isomorphisms between $F^\prime \circ F$ and the identity functor on ${\tt C}$ and between $F\circ F^\prime$ and the identity functor on ${\tt C}^\prime$. Two equivalences of categories will appear in our discussion: one between the category of sheaves and \'etal\'e spaces and the other between a category of homogeneous bundles, and the category of manifolds with an action of a given Lie group.

We will also need various incarnations of the notion of a limit. To do so, we first need to define diagrams and cones. A {\em diagram} ${\tt D}$ in the category ${\tt C}$ is a collection of objects $\{D_i\}_{i\in I}$ and morphisms $\{g_a:D_i\rightarrow D_j\}_{a\in I^\prime}$ between them.\footnote{Equivalently, a diagram is a functor from an indexing category to {\tt C}.} A {\em cone} of ${\tt D}$ is a tuple $(C,\{f_i\}_{i\in I})$ containing an object $C$ and morphisms $f_i:C\rightarrow D_i$, such that for each $g_a:D_i\rightarrow D_j$ the diagram (which really is a diagram, in the sense of our definition)
\begin{equation}
\begin{tikzcd}
& C\arrow{dr}{f_j}\arrow[swap]{dl}{f_i} &\\
D_i \arrow[swap]{rr}{g_a} && D_j
\end{tikzcd}
\end{equation}
commutes.\footnote{The fact that this diagram commutes means that to uniquely specify a cone, we do not need to specify all morphisms $f_i$, since some can be deduced. In what follows, we shall only write down those morphisms which can not be deduced from commutative diagrams.} A {\em limit} of ${\tt D}$ is a cone $(C,\{f_i\}_{i\in I})$ of ${\tt D}$ that is universal in the sense that any other cone $(C^\prime,\{f_i^\prime\}_{i\in I})$ of ${\tt D}$ factors through it via a unique {\em mediating morphism} $u:C^\prime\rightarrow C$. In other words, $f_i^\prime=f_i\circ u$ for all $i\in I$. A limit need not exist for a given diagram (and much of our work will amount to showing that they do in specific cases), but if it does it is guaranteed to be unique up to unique isomorphism. It is therefore common to abuse terminology and talk about `the' limit of a diagram, and we will do so too.

For example, a {\em pullback} is the limit of the diagram
\begin{equation} \label{eq:pull_back_limit}
\begin{tikzcd}
D_1\arrow{r}{g_1} & D_0 & \arrow[swap]{l}{g_2} D_2
\end{tikzcd}.
\end{equation}
It exists in {\tt Top} and the limiting object $C$ is given by the set $D_1 \times_{D_0} D_2 := \{(d_1,d_2) \in D_1 \times D_2| g_1(d_1) = g_2(d_2) \in D_0\}$, with the subspace topology, and the maps $f_{1,2}$ given by the restrictions to $D_1 \times_{D_0} D_2$ of the projections $D_1 \times D_2 \rightrightarrows D_{1,2}$. It does not exist, in general,  in {\tt Man} or the related categories we will consider. It does, however, exist in  {\tt Man} when one morphism, $g_2$ say, is either a surjective submersion or an open embedding, in which case $f_1$ enjoys the same property. 

A special case of a pullback is an {\em inverse image}, in which one morphism, $g_2$ say, is a monomorphism. In {\tt Set}, this is the usual inverse image and so it is common to denote the limiting object by $g_1^{-1}(D_2)$, with the other data often left implicit. As for a general pullback, the inverse image is not guaranteed to exist in  {\tt Man} or the related categories we will consider. Though, as we have seen, it does exist in  {\tt Man} in the special case where $g_2$ is not just monic but is an open embedding. Another case where it exists is the limit of $Y \xrightarrow{\pi}  X  \xleftarrow{x} * $,
where $Y$ is a fibred manifold and $x:\ast \rightarrow X$ is the inclusion of a point at $x \in X$. Here $x$ is monic, but is not an open embedding, but the limit nevertheless exists because the map $\pi$ is a surjective submersion, the limiting object being precisely the manifold given by the fibre $\pi^{-1}(x)$.

As another example, the {\em equaliser} is the limit of the diagram
\begin{equation}
\begin{tikzcd}
D_1 \arrow[shift right,swap]{r}{ g_1}\arrow[shift left]{r}{g_2} &D_0
\end{tikzcd}.
\end{equation}
Equalisers always exist in {\tt Top}, but like pullbacks may not exist in {\tt Man} or its cousins. 
\subsection{Categorical constructions}
The categories {\tt Set}, {\tt Man}, and {\tt Top} that we have introduced so far will play only a supporting r\^{o}le in our story. The main character will be a category of fibred manifolds over a fixed base, which we now define.
\begin{mydef}
Given a smooth base manifold $X$, let ${\tt Fib}_X$ denote the {\em category of fibred manifolds over $X$}, whose objects are fibred manifolds $(Y,\pi)$, where $Y$ is a smooth manifold and $\pi:Y\rightarrow X$ is a smooth surjective submersion.
A morphism, called a  {\em fibred morphism}, between objects $(Y,\pi)$ and $(Y^\prime,\pi^\prime)$ is a smooth map $f:Y\rightarrow Y^\prime$ such that $\pi^\prime\circ f=\pi$.
\end{mydef}
We will omit the adjective smooth in what follows, unless there is a risk of confusion. 

Along with ${\tt Fib}_X$, we will need a variety of other categories, defined as follows.
Let $\mathcal{O}_X$ be the category whose objects are open subsets of $X$, and whose morphisms are the inclusions of subsets. We then have the usual category ${\tt Pre}_X$ of {\em presheaves on $X$} given by the functor category ${\tt Set}^{\mathcal{O}_X^{op}}$, together with its full subcategory ${\tt She}_X$ of {\em sheaves on $X$} 
whose objects are those presheaves satisfying the sheaf condition. Finally, we need the category ${\tt Eta}_X$ of {\em \'etal\'e spaces on $X$}, an object of which is an \'etal\'e space $(E,p)$, consisting of a topological space $E$ and a local homeomorphism $p:E\rightarrow X$, and a morphism of which, called an {\em \'etal\'e morphism} is a continuous map $f:E\rightarrow E^\prime$ such that $p^\prime \circ f=p$. There is a functor ${\tt Pre}_X\rightarrow {\tt Eta}_X$ whose restriction to ${\tt She}_X$ forms, together with the functor which sends an \'etal\'e space to its sheaf of sections, an equivalence of categories (see \emph{e.g.}~\cite{Wedhorn2016}). Thus we are free to work either with ${\tt  She}_X$, or ${\tt Eta}_X$ and we will see that the latter is mainly convenient for our purposes.

Having introduced the necessary categories, we now consider functors between them. In \S\ref{subsec:Motivating_Ideas} we saw how to construct both an \'etal\'e space and the $r$th jet manifolds, using the local sections of a fibred manifold. Unsurprisingly, these constructions are functorial.
\begin{mydef}
The {\em local sections functor} $\Gamma:{\tt Fib}_X\rightarrow {\tt Eta}_X$ sends a fibred manifold $(Y,\pi)$ to the \'etal\'e space $(\Gamma Y,\Gamma\pi)$ and sends a fibred morphism $f:Y\rightarrow Y^\prime$ to the \'etal\'e morphism $\Gamma f:\Gamma Y\rightarrow \Gamma Y^\prime:[\alpha]_x\mapsto [f\circ \alpha]_x$.
\end{mydef}
\begin{mydef}
The {\em $r$th-jet functor} $J^r:{\tt Fib}_X\rightarrow {\tt Fib}_X$ sends a fibred manifold $(Y,\pi)$ to $(J^rY,\pi^r)$ and sends a fibred morphism $f:Y\rightarrow Y^\prime$ to $J^rf:J^rY\rightarrow J^rY^\prime:j^r_x\alpha \mapsto j^r_x(f\circ \alpha)$.
\end{mydef}

The functors $\Gamma$ and $J^r$ are well-behaved with respect to special classes of morphisms, as the following two theorems show. 
\begin{mylem}
\label{th:Things_preserved_by_gamma} The functor  $\Gamma$ sends an injection to an open topological embedding, but does not necessarily send surjections to surjections.  (Proof: Appendix~\ref{ap:Things_preserved_by_gamma})
\end{mylem}
\begin{mylem}\label{th:Structures_preserved_by_jet_functor} The functor $J^r$ preserves submersions, surjective submersions, immersions, injective immersions, and embeddings, but does not necessarily preserve surjections or injections. (Proof: Appendix~\ref{ap:Structures_preserved_by_jet_functor}) 
\end{mylem}

Finally we introduce two sets of natural transformations involving $\Gamma$ and $J^r$, obtained either by forgetting the derivatives of sections or by prolonging sections to higher-jet manifolds.
\begin{mydef}
For $r \ge l\ge 0$, the {\em forget derivatives map} is the natural transformation $J^r\Rightarrow J^l$ defined on $(Y,\pi)$ by the
surjective submersion (in fact, affine bundle map for $l \geq r-1$) $\pi^{r,l}:J^rY\rightarrow J^lY: j^r_x\alpha\mapsto j^l_x\alpha$. 
\end{mydef}
\begin{mydef}
For $r>0$, the  {\em prolong sections map} is the natural transformation $\Gamma\Rightarrow \Gamma J^r$ defined on $(Y,\pi)$ by  $j^r:\Gamma Y\rightarrow \Gamma J^r Y:[\alpha]_x\mapsto [j^r\alpha]_x$, where $[\alpha]_x$ is the germ at $x$ of the local section $\alpha$ on $U \ni x$.
\end{mydef}
\section{Constraints} \label{sec:Constraints}
\subsection{Holonomic and higher-degree constraints}
In the physicist's world of local coordinates $(x^\mu,y^a)$, a holonomic constraint is usually defined as a set of smooth relations of the form  $f(x^\mu,y^a)=0$. The inadequacy of this definition can easily be seen by considering examples from classical mechanics in the plane (so $\pi: Y \to X$ is the map $\R^3 \to \R: (x^0,y^1,y^2) \mapsto x^0$), such as $y^1y^2 = 0$ or $(y^1)^2 + (y^1)^2 + (x^0)^2 - 1 = 0$. Ills of the kind observed in the first example can be cured by insisting that a holonomic constraint be an embedded submanifold $Z$ of $Y$ and those in the second example by insisting that $Z$ itself be a fibred manifold over $X$, embedded in $Y$ via a fibred morphism \cite{Krupkova09,Krupkova00,Krupkova97}. Thus we make the following
\begin{mydef} \label{def:fibred_subman}
A {\em fibred submanifold} (resp. {\em open fibred submanifold}) of a fibred manifold $(Y,\pi)$ is a fibred manifold $(Z,\zeta)$ together with a fibred morphism $\iota_Z:Z\rightarrow Y$ that is an embedding (resp. open embedding).
\end{mydef}
A holonomic constraint as defined in \cite{Krupkova09,Krupkova00,Krupkova97} then amounts to a choice of fibred submanifold of $(Y,\pi)$ and we will use this as a working definition for now (later we will make an equivalent definition that appears rather perverse, but turns out to be much more useful for finding more general constraints).
The local degrees of freedom of the field theory can then obviously be taken to be the local sections of $(Z,\zeta)$. Since these form a sheaf whose stalks are non-empty (since $\zeta$ is a surjective submersion), we obtain a theory which is consistent with locality and in which local degrees of freedom exist. 

At some level, this corresponds to the physicist's notion that holonomic constraints are easily dealt with, because one can simply eliminate redundant degrees of freedom. But it is important to note that our working definition of 
a holonomic constraint is much more than just a coordinate independent reformulation of the usual physicist's notion. Not only does it remove pathological examples such as those already discussed, but it also includes constraints which would be considered nonholonomic by the physicist, in that they cannot be expressed locally in terms of relations $f(x^\mu,y^a)=0$. For example, in classical mechanics in the plane, our working definition includes the fibred submanifold defined by $(y^1)^2 + (y^2)^2 > 1$.

Now let us turn our attention to constraints which are nonholonomic in the sense that they include derivatives of order $r>0$ and below of the fields, in local coordinates. An obvious guess is to consider a fibred submanifold not of $(Y,\pi) \cong (J^0Y,\pi^0)$, but rather of $(J^rY,\pi^r)$. Denoting the fibred morphism embedding by $\iota_Q:Q\rightarrow J^rY$, the degrees of freedom of the field theory would then correspond to the local sections of $(Y,\pi)$ whose prolongation to $J^rY$ lies in $\iota_Q(Q) \subset J^rY$. We now encounter two potential difficulties. One is that it is not obvious,  {\em a priori}, that the constraint is consistent with locality, in that the degrees of freedom form a sheaf. 
Even if they do, it is not obvious that degrees of freedom exist at every spacetime point in $X$, or in other words that the stalks of the sheaf are not empty. In fact, it will turn out that the first condition is automatically satisfied, but this will require some work to show, so let us return to it shortly. The second condition is not automatically satisfied, as the following counterexample from classical mechanics in the plane shows. The first jet manifold there is given by $(J^1Y,\pi^1)=(\mathbb{R}^5,(x^0, y^1,y^2, y^1_0,y^2_0)\mapsto x^0)$; letting $(Q,\nu)=(\mathbb{R}^3,(x^0,y^1,y^1_0)\rightarrow x^0)$ with $\iota_Q:(x^0,y^1,y^1_0)\mapsto ( x^0, y^1 , 0,y^1_0,1)$, we see that there are no local sections at all!

Now let us return to the first condition. The statement that the degrees of freedom form a sheaf is equivalent to the following
\begin{myth} \label{th:Existence_of_sheaf_pullback}
The pull-back of  $\Gamma \iota_Q:\Gamma Q\rightarrow \Gamma J^rY$, and $j^r:\Gamma Y\rightarrow \Gamma J^r Y$ in ${\tt Eta}_X$ exists  and we denote it by $(E^Q,p^Q)$. (Proof: Appendix~\ref{ap:Existence_of_sheaf_pullback})
\end{myth}

These considerations motivate the following
\begin{mydef} \label{def:consistent_constraint}
A {\em consistent constraint of order $r$} on the fibred manifold $(Y,\pi)$ is a subfibred manifold $Q \subset J^rY$ such that the stalks of the pullback $E^Q$, whose existence was shown in the previous theorem, are non-empty.
\end{mydef}

The difficulty with nonholonomic constraints, at least those defined by a submanifold $Q \subset J^rY$, thus reside in establishing that the stalks are non-empty. The rest of this Section will be devoted to finding ways in which this can be achieved. 

To do so, it is useful to re-examine the notion of a holonomic constraint, our working definition of which identifies it with a consistent constraint of order $0$. The following argument shows, however, that we are also free to regard it as a consistent constraint of any order $r$. Firstly, Lemma~\ref{th:Structures_preserved_by_jet_functor} has shown that the functor $J^r$ sends a subfibred manifold 
$\iota_Z:Z\rightarrow Y$ to a subfibred manifold $J^r\iota_Z:J^rZ\rightarrow J^rY$. Moreover, the resulting \'etal\'e spaces $(E^Q,p^Q)$ are isomorphic (to $(\Gamma Z,\Gamma \zeta)$) for all $r$, so define consistent constraints of order $r$ that lead to field theories with equivalent degrees of freedom.

The notion of  different constraints leading to theories that are physically the same, in the sense of having equivalent degrees of freedom, leads us to make the following
\begin{mydef}
Consistent constraints (of any order) are {\em kinematically equivalent} if their corresponding \'etal\'e spaces are isomorphic.
\end{mydef}
Going further, let us make the following, apparently rather perverse, definition of a holonomic constraint.
\begin{mydef} \label{def:holonomic}
A {\em holonomic constraint of degree $r$ for $(Z,\Omega)$} is a limit in ${\tt Fib}_X$ of the diagram
\begin{equation} \label{eq:commutator_holonomic}
\begin{tikzcd}
J^rZ \arrow[swap]{d}{ \zeta^{r,0}} \arrow{r}{J^r\iota_Z} \arrow{dr}{\Omega} & J^rY\arrow{d}{\pi^{r,0}} \\
 Z\arrow[swap]{r}{\iota_Z} & Y
\end{tikzcd}
\end{equation}
where $(Z,\zeta)$ is a fibred submanifold of $(Y,\pi)$ with embedding $\iota_Z$ and the fibred morphism $\Omega$ is such that the lower triangle commutes (along with the square).
\end{mydef}
The definition is perverse for more than one reason. Firstly, the requirement that the lower triangle commutes evidently shows that given $(Z,\iota_Z)$ there exists a unique map $\Omega$, namely $\iota_Z \circ \zeta^{r,0}$, so there is no data associated to $\Omega$. Secondly, the fact that the square commutes shows that the limiting object is (uniquely isomorphic to) $(J^rZ,\zeta^r)$, with the fibred morphism to $J^rZ$ in the diagram being the identity and with all other fibred morphisms being fixed by the commutativity of the diagram. Nevertheless, it is clear that our new definition is equivalent to our old working definition, in that it yields a kinematically equivalent constraint. 

The beauty (if it can be called that) of our new definition is that admits a non-trivial dual, to which we now turn.

\subsection{Coholonomic constraints} \label{subsec:coholonomic}
We begin with a preliminary definition that is the dual of \ref{def:fibred_subman}.
\begin{mydef}
A {\em fibred quotient of  the fibred manifold $(Y,\pi)$} is a fibred manifold $(Z,\zeta)$ together with a fibred morphism $\tau_Z : Y \to Z$ that is a surjective submersion.
\end{mydef}
Dualising our new definition of a holonomic constraint, we have the following
\begin{mydef}\label{def:coholonomic}
A {\em coholonomic constraint of degree $r$ for $(Z,\Omega)$} is a limit in ${\tt Fib}_X$ of the diagram
\begin{equation} \label{eq:commutator_coholonomic}
\begin{tikzcd}
J^rY \arrow[swap]{d}{ \pi^{r,0}} \arrow{r}{J^r\tau_Z} & J^rZ\arrow{d}{\zeta^{r,0}} \arrow[swap]{dl}{\Omega}\\
 Y\arrow[swap]{r}{\tau_Z} & Z
\end{tikzcd}
\end{equation}
where $(Z,\zeta)$ is a fibred quotient of $(Y,\pi)$ whose surjective submersion is $\tau_Z$ and the fibred morphism $\Omega$ is such that the lower triangle commutes (along with the square).
\end{mydef}

A number of remarks are now in order. Firstly, we remark that our `dual' construction is not obtained by dualising willy-nilly. Rather, we simply replace the notion of a fibred submanifold, namely a fibred manifold together with a fibred morphism from it to 
$(Y,\pi)$ that is an embedding, by the dual notion of a fibred quotient. We have not changed the direction of the map $\Omega$, and nor have we replaced the limit by a colimit. 

Secondly, we remark that the map $\Omega$, which now takes the form of a lift of $\zeta^{r,0}$ through $\tau_Z$, no longer necessarily exists; nor, if it does, is it necessarily unique. As we shall see, this opens the door to a rather rich notion of a constraint, which will capture, in particular, the essence of the inverse Higgs phenomenon.  

Thirdly, we remark that if we were to remove the datum of the map $\Omega$ from the definition, we would not obtain anything interesting. The limit in that case is simply $J^rY$, so we recover the unconstrained field theory on $Y$. 

A fourth remark is that it is not obvious that the limit we have defined exists. In fact we have the following
\begin{myprop} \label{th:limit_coholonomic_no_group_action}
The limit of Diagram~\ref{eq:commutator_coholonomic} exists; denoting it by $((Q,\nu),\{\iota_Q:Q\rightarrow J^rY\})$, $\iota_Q$ is an embedding. (Proof: Appendix~\ref{app:Proof_Limits_Coho_Comer})
\end{myprop}
Because $Q$ is embedded, we are furthermore guaranteed, by Theorem \ref{th:Existence_of_sheaf_pullback} above, that the degrees of freedom form a sheaf, so are consistent with locality. But in fact much more is true.  
\begin{myth} \label{th:isomorphism_of_sheaves_coholonomic}
The \'etal\'e space $(E^Q,p^Q)$ for a coholonomic constraint of degree $r$ for $(Z,\Omega)$ is isomorphic to $(\Gamma Z,\Gamma \zeta)$. (Proof: Appendix~\ref{ap:Isomorphism_unconstrained_constrained})
\end{myth}
So not only are coholonomic constraints of degree $r$ for $(Z,\Omega)$ consistent constraints, but, just as for holonomic constraints, we find that they are kinematically equivalent to the unconstrained theory on the fibred manifold $Z$.  Comparing with the physics literature, we see that our theorem corresponds to the notion of `essential Goldstone bosons'. Indeed, these are to be interpreted precisely as the local description in adapted coordinates of the  local sections of $\zeta: Z \to X$.

Moreover, our theorem shows that, even though we started from a definition of coholonomic constraint which was not the exact categorical dual of a holonomic constraint, we end up with a duality at the level of field theories which is satisfyingly precise: a holonomic constraint is kinematically equivalent to an unconstrained theory on a fibred submanifold, while a coholonomic constraint is kinematically equivalent to an unconstrained theory on a fibred quotient. 
\subsection{Meronomic constraints} \label{subsec:meronomic}
To describe all of the examples of the inverse Higgs phenomenon in the literature within our formalism requires us to slightly generalise the notion of coholonomic constraints. This is most conveniently done by first generalising holonomic constraints and then dualising as before.

Locally, meronomic constraints look like holonomic constraints and so we call them {\em meronomic} constraints (from the greek for `part' and `law', in much the same way that holonomic is from `whole' and `law'). Compared with holonomic constraints, we have an extra datum in the form of an open fibred submanifold of $J^rZ$.
\begin{mydef}\label{def:meronomic}
A {\em meronomic constraint of degree $r$ for $(Z,R,\Omega)$} is a limit in ${\tt Fib}_X$ of the diagram
\begin{equation}\label{eq:commuting_meronomic}
\begin{tikzcd}[column sep=2cm]
R\arrow[swap]{d}{\iota_R} \arrow[pos=0.3]{ddr}{\Omega}&\\
J^rZ\arrow[swap]{d}{\zeta^{r,0}} \arrow[crossing over,pos=0.7]{r}{J^r\iota_Z} & J^rY \arrow{d}{\pi^{r,0}}\\
Z\arrow[swap]{r}{\iota_Z} & Y
\end{tikzcd}
\end{equation}
where $(Z,\zeta)$ is a fibred submanifold of $(Y,\pi)$ with embedding $\iota_Z$, $(R,\rho)$ is an open fibred submanifold of $(J^rZ,\zeta^r)$ with open embedding $\iota_R$, and the fibred morphism $\Omega$ is such that the diagram commutes.
\end{mydef} 

Just as for holonomic constraints, the datum of the map $\Omega$ adds nothing here, since it must equal $\iota_Z \circ \zeta^{r,0} \circ \iota_R$, but is present so that we obtain something more general when we dualise.\footnote{Amusingly, if we dualise without the map $\Omega$, we obtain not a trivial unconstrained theory (as we did in the holonomic case), but rather a class of constraints that are equivalent to a subclass of meronomic constraints. This fact is proven and made use of in Theorem \ref{th:limit_comeronomic_no_group_action}.}

Completely analogously to a holonomic constraint, the limit in the definition exists and is given by $((R,\rho),\{\id:R\rightarrow R\})$, up to unique isomorphism. 

The fact that $\iota_R$ is an open embedding is what makes a meronomic constraint locally look like a holonomic constraint. Due to this, the \'etal\'e space $(E^Q, p^Q)$ is guaranteed to have non-empty stalks, since, roughly, for any $x\in X$ there will be a local section of $Z$, $\beta$ with $J^r\beta(x)$ lying in the open set $R$, we can then just restrict the domain of $\beta$ so that $J^r\beta$ lies wholly in $R$. $[J^r\beta]_x$ then defines a point in $(p^Q)^{-1}(x)$.

We recover the special case of a holonomic constraint by choosing $\iota_R$ to be an isomorphism.
\subsection{Comeronomic constraints} \label{subsec:comeronomic}
Turning the handle, we now obtain the dual notion corresponding to a meronomic constraint, which is relevant for certain physical examples of the inverse Higgs phenomenon.
\begin{mydef} \label{def:comeronomic}
A {\em comeronomic constraint of degree $r$ for $(Z,R,\Omega)$} is a limit in ${\tt Fib}_X$ of the diagram
\begin{equation} \label{eq:Commutative_Comeronomic}
\begin{tikzcd}[column sep=2cm]
& R\arrow[swap,pos=0.3]{ddl}{\Omega} \arrow{d}{\iota_R}   \\
J^rY \arrow[swap]{d}{ \pi^{r,0}} \arrow[crossing over,pos=0.3]{r}{J^r\tau_Z} & J^rZ\arrow{d}{\zeta^{r,0}} \\
Y\arrow[swap]{r}{\tau_Z} & Z
\end{tikzcd}
\end{equation}
where $(Z,\zeta)$ is a fibred quotient of $(Y,\pi)$ whose surjective submersion is $\tau_Z$, $(R,\rho)$ is an open fibred submanifold of $(J^rZ,\zeta^r)$ with open embedding $\iota_R$, and the fibred morphism $\Omega$ is such that the  diagram commutes.
\end{mydef}
\begin{myprop} \label{th:limit_comeronomic_no_group_action}
The limit of the Diagram~\ref{eq:Commutative_Comeronomic} exists, denoting it by $((Q,\nu),\{\iota_Q:Q\rightarrow J^rY,f_Q^R:Q\rightarrow R\})$, then $\iota_Q$ is an embedding. (Proof: Appendix~\ref{app:Proof_Limits_Coho_Comer})
\end{myprop}

As with holonomic constraints and meronomic constraints, a coholonomic constraint is a special instance of a comeronomic constraint, corresponding to the case where $\iota_R$ is an isomorphism. 

For a holonomic constraint we had that the \'etal\'e space $(E^Q,p^Q)$ was isomorphic to $(\Gamma Z,\Gamma\zeta)$. For comeronomic constraints we have the following
\begin{myth} \label{th:isomorphism_of_sheaves_comeronomic}
The \'etal\'e space $(E^Q,p^Q)$ associated with a comeronomic constraint is isomorphic to the \'etal\'e space $( E^R, p^R)$ associated with the embedding of $R$ into $J^rZ$. (Proof: Appendix~\ref{ap:Isomorphism_unconstrained_constrained})
\end{myth}

For the same reason that meronomic constraints lead to non-empty stalks and hence consistent constraints, the stalks of $( E^R, p^R)$ will be non-empty and, due to the isomorphism, so will those of $(E^Q,p^Q)$. Even more importantly, we learn that a comeronomic constraint is kinematically equivalent to a meronomic constraint on the fibred quotient $(Z,\zeta)$.
\subsection{An example from classical mechanics: the Chaplygin sleigh} 
Here we give an example of a comeronomic constraint in classical mechanics, showing that, despite their abstract definition, they occur in remarkably simple examples. The example is based on the famous example of a Chaplygin sleigh, with the minor tweak that we forbid the sleigh from being translationally at rest, thus deleting a single point from the space of possible translational velocities of the sleigh.

Recall that a Chaplygin sleigh is a rigid body sliding in the plane, with motion that is frictionless apart from a `knife edge' at a point on the object that prevents motion at that point perpendicular to the edge of the knife, as in Fig.~\ref{fig:Chaplygin_Sleigh}. 
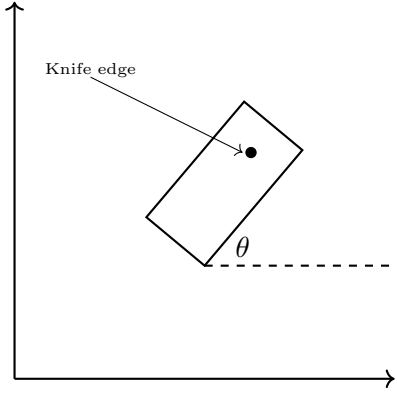
\begin{figure}
\begin{center}
\begin{tikzpicture}
\coordinate (O) at (0,0);
\coordinate (y_end) at (0,5);
\coordinate (x_end) at (5,0) {};
\draw[->,thick] (O) -- (y_end);
\draw[->,thick] (O) -- (x_end);
\draw[draw=black,rotate around={-40:(2.5,1.5)},thick] (1.5,1.5) rectangle ++(1,2); 
\node[circle,fill,inner sep=1.5pt] at (3.11,3) {};
\draw[dashed,thick] (2.5,1.5) -- (5, 1.5);
\node at (3,1.75) {$\theta$};
\draw[->] (1,4) -- (3,3);
\node at (1,4.1) {\tiny Knife edge};
\end{tikzpicture}
\end{center}
\caption{The Chaplygin sleigh}
\label{fig:Chaplygin_Sleigh}
\end{figure}

The fibred manifold $Y$ over $\R$ is thus $\mathbb{R}^3\times S^1$  with local adapted coordinates $(t,x,y,\theta)$ representing the time, position of the knife edge in the plane, and orientation of the sleigh, and with fibering map $(t,x,y,\theta) \mapsto t$. The jet bundle $J^1Y$ is thus  $\mathbb{R}^3\times S^1 \times \R^3$  with local adapted coordinates $(t,x,y,\theta,x_t,y_t,\theta_t)$. 
To describe the system as a comeronomic constraint, we start with the fibred quotient of $Y$ obtained by projecting out the $S^1$, which admits global coordinates $(t,x,y)$ and consider the open fibred submanifold $R$ of $J^1Z = \{(t,x,y,x_t,y_t)\}$ obtained by deleting the points with $x_t = y_t = 0$ (enforcing the constraint that the sleigh is not allowed to be translationally at rest). This allows us to define a fibred morphism $\Omega: R \to Y$ which acts as the identity on $(t,x,y)$ but sends $(x_t,y_t)$ to the point $(x_t/\sqrt{x_t^2 + y_t^2}, y_t/\sqrt{x_t^2 + y_t^2})$ on the unit circle $S^1 \subset \R^2$.  This has precisely the effect of enforcing the constraint that the sleigh may not move perpendicularly to the knife edge at the knife edge. 

As above, this theory is kinematically equivalent to a theory with a meronomic constraint defined by $R$ embedded into $Z$; the explicit isomorphism takes the stalk whose section is defined by $t\mapsto (t,x(t),y(t),\theta(t))$, (where $x(t)$, $y(t)$ and $\theta(t)$ are required to satisfy the constraint), to the stalk defined by the section $t\mapsto (t,x(t),y(t))$ in $E^R$.
\subsection{Summary of the classes of constraint}
To close our discussion of constraints in this Section, we provide a summary of what we have done. 
We began by observing that constraints in physics cannot be chosen arbitrarily, but rather 
must not violate the basic tenets of locality, and moreover must be such that an observer at any location must have something to observe. Our definition of consistent constraints in \ref{def:consistent_constraint} takes these basic principles into account. 

It is usual in physics not to worry about these consistency conditions, but they may be present even with the type of constraint that is most often studied, namely a holonomic constraint, at least as it is usually defined.  (A simple example of such a constraint is a particle in a plane whose motion is restricted to lie on the circle $(y^1)^2+(y^2)^2=1$.) Our inequivalent definition in \ref{def:holonomic} not only excludes inconsistent constraints, but is also marginally broader, containing, for example, constraints of the form  $(y^1)^2+(y^2)^2>1$.

One has to be more careful about the consistency conditions when one considers constraints that are non-holonomic (in the sense that they don't satisfy the conditions of our Definition \ref{def:holonomic}). One way to guarantee consistency is via the 
generalisation of a holonomic constraint that we termed a meronomic constraint in \ref{def:meronomic}. Being more general, such constraints may be either holonomic or non-holonomic, but locally they all appear holonomic (in the sense of Definition \ref{def:holonomic}), so this is not much of a generalization. An example of such a constraint is the condition $\partial_t y^1>0$.

Next we were led, with the aim of studying inverse Higgs phenomena, to two more classes of non-holonomic constraint, namely 
coholonomic (Definition \ref{def:coholonomic}) and comeronomic (Definition \ref{def:comeronomic}) constraints. An example of a comeronomic constraint was detailed in the previous Subsection, in the form of the Chaplygin sleigh. From the point of view of local co-ordinates they appear somewhat radical, but from the category-theoretic point of view they are seen to be `merely' duals of 
holonomic constraints and meronomic constraints respectively. This duality not only guarantees their consistency, but also implies that, just as meronomic constraints are a generalisation of holonomic constraints, so too are comeronomic constraints a generalisation of coholonomic constraints.

Pictorially, the interrelationships between the different types of constraints are as follows:
\begin{equation}
\begin{tikzcd}[remember picture]
\text{Holonomic} \arrow[r,rightsquigarrow,"dual"]& \text{Coholonomic} \\
\text{Meronomic} \arrow[r,rightsquigarrow,swap,"dual"]& \text{Comeronomic}
\end{tikzcd}
\begin{tikzpicture}[overlay,remember picture]
\path (\tikzcdmatrixname-1-1) to node[midway,sloped]{$\subset$}
(\tikzcdmatrixname-2-1);
\path (\tikzcdmatrixname-1-2) to node[midway,sloped]{$\subset$}
(\tikzcdmatrixname-2-2);
\end{tikzpicture}.
\end{equation}
\section{Fibrewise group actions and homogeneous bundles} \label{sec:homogeneous}
Having described the consistent constraints that appear in theories featuring the inverse Higgs phenomenon, we now discuss the r\^{o}le played by symmetry, in the form of a Lie group $G$ acting smoothly on $Y$. Things simplify greatly in the case where $G$ also acts on $X$ such that the fibering map $\pi:Y \to X$ is $G$-equivariant, for the simple reason that a well-defined group action is then induced on each $r$-jet manifold $J^rY$, and this action is such that the maps $\pi^{r,l}$ are $G$-equivariant. We call such an action a {\em fibrewise group action}. For more general $G$ actions on $Y$, one induces at best a partial group action on $J^rY$ and we will defer the somewhat technical study of this situation to the next Section.

When the $G$ action is fibrewise, it is possible to define a number of subgroups of $G$ that are familiar to physicists (it is important to remark that none of these subgroups are defined in the case of more general group actions). 
For each $x\in X$ we define the {\em internal symmetry group at $x$}, as the stabiliser $G_x$ of $x \in X$. The {\em internal symmetry group} $G_X$ can then be defined as $\cap_x G_x$; equivalently, $G_X$ is the subgroup of $G$ that acts trivially on $X$. $G_X$ is a normal subgroup of $G$, and we can define the \emph{spacetime symmetry group} as the group $G/G_X$.

Of most interest to us (since we are interested in theories of Goldstone bosons) is the case where $G$ acts, in addition, transitively on $Y$, such that $Y$ is diffeomorphic to $G/K$ for some Lie subgroup $K\subseteq G$. Because $\pi$ is surjective and $G$-equivariant, it follows that $G$ also acts transitively on $X$, so we have that $X$ is diffeomorphic to $G/H$ for some Lie subgroup $H \subseteq G$ such that $H \supseteq K$. Moreover, the fibred manifold $\pi:Y \to X$ is isomorphic (in ${\tt Fib}_X$) to $G/K \to G/H$, which has the structure of a fibre bundle with fibre $H/K$ associated to the $H$-principal bundle $G \to G/H$. This, along with the corresponding jet manifolds, is an example of a {\em homogeneous bundle} and the theory of such bundles can be brought to bear.

To give a simple example that allows us to make contact with the typical situation encountered in physical theories,
suppose that $G=A\times B$ for some Lie groups $A$ and $B$, and let $K\subseteq B$ and $H=A\times K$, so that $K \subseteq H \subseteq G$ as required. Recalling that 
$Y\cong G/K \cong A \times B/K$ and $X\cong G/H\cong B/K$, we have that the internal symmetry group at $bK\in X$ is $A\times K_{bK}$, where $K_{bK}$ is the subgroup of $K$ given by $\{k \in K|  kbK=bK\}$. If $K$ is, say, the Lorentz group and $B$ the Poincar\'e group, we have that $ G_X = A$ and $G/G_X=B$. In other words, the internal symmetry is $A$ and the spacetime symmetry is $B$. We stress that this simple result will not obtain in more general situations, even when the $G$ action is fibrewise. 

We now wish to go further and discuss the 
group actions that are induced on jet manifolds and their interplay with coholonomic and comeronomic constraints. A first observation is that, even if we start with a transitive group action on $Y$, for sufficiently large $r$ the group action induced on $J^rY$ will not be transitive. Indeed, since a manifold with a transitive action of $G$ is diffeomorphic to a homogeneous space of $G$, the dimension of such a manifold is bounded above by the dimension of $G$. But the dimension of $J^rY$ increases without bound with $r$. It it is this simple fact that both allows for, and exhibits the generic nature of, the inverse Higgs phenomenon: once we include enough derivatives in a field theory, $G$ cannot act transitively and subsets of the orbits $G$ can be used to define non-trivial constraints that are nevertheless compatible with the action of $G$. Since they necessarily involve derivatives ($G$ acts transitively on $Y \cong J^0Y$, so there are no constraints that are compatible with the $G$ action)
the constraints are necessarily nonholonomic, according to the usual definition, leading to possible problems with consistency. But all constraints in the literature on the inverse Higgs phenomenon turn out to be either coholonomic or comeronomic, so consistency is guaranteed. 

To explore this in more detail requires us to first review the theory of homogeneous bundles based on the principal $L$-bundle $G \to G/L$, for $L \subseteq G$. The key observation here is that these form a category and that that category is equivalent to the category of manifolds with an $L$ action. This equivalence of categories is a rigorous statement of the physicist's vague notion that, in sigma models, $G$ invariance follows from $L$ invariance alone.

Some of the discussion in this Section requires results extending the results of the previous Section to the case where a group acts. The proofs of these results are subsumed into the proofs for the more general case of a partial group action, given in the next Section and the Appendices.
\subsection{The category of homogeneous bundles}
We now review the theory of homogeneous bundles. For more details, see {\em e.g.}~\cite{CapSlovak_2009}.

Let $G$ be a Lie group, and $L$ a Lie subgroup of $G$. A {\em homogeneous bundle over the homogeneous space $G/L$} is a triple $(Y,\pi,\mathcal{Y})$ consisting of a smooth manifold $Y$ equipped with a smooth action $\mathcal{Y}: G \times Y \to Y$ of $G$ and a smooth bundle map $\pi: Y \to G/L$ that is equivariant with respect to $\mathcal{Y}$ and the usual action  $\mathcal{L}:G \times  G/L\rightarrow G/L$ of $G$ given by $\mathcal{L}_g: g^\prime L\mapsto gg^\prime L$.

The homogeneous bundles over $G/L$ form the objects of a category, which we now define.
\begin{mydef}
Let ${\tt HBun}_{G/L}$ be the category whose objects are homogeneous bundles over the homogeneous space $G/L$, and whose morphisms from $(Y,\pi,\mathcal{Y})$ to $(Y^\prime,\pi^\prime,\mathcal{Y}^\prime)$ are smooth maps $f:Y\rightarrow Y^\prime$ such that $\pi^\prime\circ f=\pi$ and $\mathcal{Y}_g^\prime\circ f=f\circ \mathcal{Y}_g$ for all $g\in G$.
\end{mydef}

The category ${\tt HBun}_{G/L}$ is equivalent to the category defined as follows.
\begin{mydef}
Let  $L\mathen{\tt Man}$ be the category whose objects are pairs $(M,\mathcal{M})$, consisting of a smooth manifold $M$ equipped with a smooth action $\mathcal{M}: L \times M \to M$ of $L$, which we call an {\em $L$-manifold}, and whose morphisms between $(M,\mathcal{M})$ and $(M^\prime,\mathcal{M}^\prime)$ are smooth maps  $f:M\rightarrow M^\prime$ such that $\mathcal{M}_l^\prime \circ f=f\circ \mathcal{M}_l$ for all $l\in L$, which we call {\em $L$-maps}.
\end{mydef}
We will not give the functors defining this equivalence, which we denote by $\Pi:{\tt HBun}_{G/L}\rightarrow L\mathen{\tt Man}$ and $\hat \Pi: L\mathen{\tt Man}\rightarrow {\tt HBun}_{G/L}$, explicitly (the reader is directed to~\cite{CapSlovak_2009} for an explicit form), but simply record the following lemma.

\begin{mylem} \label{th:homogeneous_preserve_open_embedding}
The functors $\Pi$ and $\hat \Pi$ send open embeddings to open embeddings. (Proof: Follows manifestly from the definitions of $\Pi,\hat \Pi,$ and the quotient and subspace topologies.)
\end{mylem}

\subsection{Constructing constraints}
To specify a comeronomic constraint with fibrewise group actions requires the following data:
\begin{enumerate}
\itemsep0em 
\item a fibred manifold $(Y,\pi)$, a fibred quotient $(Z,\zeta)$ of $(Y,\pi)$, an open fibred submanifold $(R,\rho)$ of $(J^r Z,\zeta^r)$, and a suitable fibred morphism $\Omega:R\rightarrow Y$  (as per the definition with no group acting given in \ref{def:comeronomic});
\item  fibrewise group actions $\mathcal{Y}$, $\mathcal{Z}$, and $\mathcal{R}$ of $G$ on $Y,Z$, and $R$ such that: the surjective submersion $\tau_Z:Y\rightarrow Z$ is equivariant with respect to the actions $\mathcal{Y}$ and $\mathcal{Z}$; the open embedding $\iota_R:R\rightarrow J^r Z$ is equivariant with respect to the action $\mathcal{R}$ and the action $J^r\mathcal{Z}$ of $G$ on $J^rZ$ induced by $\mathcal{Z}$; \footnote{In adapted local coordinates, this action can be deduced using the chain rule; \S\ref{sec:partial_actions} gives a formal definition.} the fibred morphism $\Omega$ is equivariant with respect to $\mathcal{R}$ and $\mathcal{Y}$. 
\end{enumerate}

Specifying this data becomes simpler in the case of most physical interest, namely when $G$ acts transitively on $Y$, where we have the following
\begin{myth}
Let a comeronomic constraint be defined by a diagram as in \ref{eq:Commutative_Comeronomic}, where all objects have a $G$-action and all morphisms are $G$-equivariant, and let $G$ act transitively on $Y$. Then $Z \cong G/L$ for some $L \subseteq G$ and all objects and morphisms in the diagram lie in the subcategory ${\tt HBun}_{G/L}$. 
\end{myth}
Thus  we can describe everything in terms of homogeneous bundles or, via the equivalence of categories, in terms of manifolds with an $L$-action. 
\begin{proof}
The maps $\tau_Z$ and $\pi$ are required to be $G$-equivariant, so it follows that $G$ also acts transitively on $Z$ and $X$, so we can write $Y \cong G/K$, $X \cong G/H$, and $Z \cong G/L$, with $K\subseteq L\subseteq H\subseteq G$. Moreover, the maps $\tau_Z$ and $\zeta$ are $G$-equivariant bundle maps and so we have that $(Y,\tau_Z)$ and $(Z,\zeta)$ define objects in ${\tt HBun}_{G/L}$ whose typical fibres are the $L$-manifolds given by $L/K$ and a point, respectively. Further, since the map $\zeta^{r,0}$ is a $G$-equivariant bundle map, we have that $(J^rZ,\zeta^{r,0})$ also defines an object in 
${\tt HBun}_{G/L}$. 

Now consider the open fibred submanifold $R$ in $J^rZ$. Because $G$ acts transitively on $Z$, the equivariant map $\zeta^{r,0} \circ \iota_R$ must be a bundle map. The argument goes as follows. Because $G$ acts transitively, the map must be a surjection and because both $\zeta^{r,0}$ and $\iota_R$ are submersions, it must also be a submersion. But then the same arguments given in \cite{CapSlovak_2009} to derive the equivalence of categories between homogeneous bundles and $L$ manifolds show that the map is isomorphic to a bundle map. (In particular, it is clear that the fibres of the fibred manifold are all diffeomorphic to one another, since any one can be reached from another by a diffeomorphism corresponding to some $g \in G$.) So $(R, \zeta^{r,0} \circ \iota_R)$ also defines an object in ${\tt HBun}_{G/L}$. All the morphisms in  Diagram \ref{eq:Commutative_Comeronomic} are equivariant by assumption and commutativity of the diagram ensures that they define morphisms in ${\tt HBun}_{G/L}$.
\end{proof}

So we can carry the discussion over to $L\mathen{\tt Man}$, where $Z$ is represented by a point and $Y$ is represented by the homogeneous space $L/K$. Suppose that $R$ is represented by the $L$-manifold $\Pi R$. 
For generic $r$, an explicit description of $J^rZ \to Z$ as an $L$-manifold is somewhat unpleasant; we content ourselves with giving a description for $r=1$ where, since $J^1Z\to Z$ is an affine bundle, we obtain an affine space with an action of $L$.
This covers all examples in the literature, bar one, corresponding to the Galileid~\cite{Nicolis:2015sra}, where one needs $r=2$.
\begin{myprop}
For a fibred manifold $Z\cong G/L \to X \cong G/H$, a typical fibre of the affine bundle $J^1Z \to Z$ is given by the $L$-affine space $A(\mathfrak{g}/\mathfrak{h},\mathfrak{g}/\mathfrak{l})$ over $\mathrm{Hom} (\mathfrak{g}/\mathfrak{h},  \mathfrak{h}/\mathfrak{l})$ of linear sections of the linear map $\mathfrak{g}/\mathfrak{l} \to \mathfrak{g}/\mathfrak{h}$,
where $\mathfrak{g}$ denotes the Lie algebra of $G$, {\em \&c.}; the action of $L$ is by pre- or post-composition with the actions on $\mathfrak{g}/\mathfrak{l}$, and $\mathfrak{h}/\mathfrak{l}$ induced by the adjoint action of $L \subseteq G$ on $\mathfrak{g}$. (Proof: Follows from \cite[Lemma 4.1.3]{Saunders_1989}.)
\end{myprop}

In all, we have the following
\begin{myth} When $G$ acts transitively on $Y$, the required data for a comeronomic constraint of order 1 can be specified by
\begin{enumerate}
\item a chain of inclusions of 4 Lie groups, $K\subseteq L\subseteq H\subseteq G$, which define $Y,Z,X$ in Diagram \ref{eq:Commutative_Comeronomic};
\item an open $L$-submanifold $\Pi R$ of $A(\mathfrak{g}/\mathfrak{h},\mathfrak{g}/\mathfrak{l})$, which defines $(R,\iota_R)$; 
\item an $L$-map $\Pi \Omega: \Pi R\rightarrow L/K$, which defines $\Omega$.
\end{enumerate}
\end{myth}
One checks that by Lemma \ref{th:homogeneous_preserve_open_embedding} we get an open embedding $\iota_R$ if and only if we start from an open embedding in $L\mathen{\tt Man}$ and that the fibred morphism $\Omega$ is such that the diagram commutes.

We now go on to describe a number of examples.
\subsection{Examples}
We will now list examples of inverse Higgs phenomena taken from the  literature~\cite{Goon:2012dy,Nicolis:2013lma,Nicolis:2015sra}. For each of the examples, we will specify all the data indicated in the previous Subsection required to specify a comeronomic constraint. In cases where the constraint is in fact coholonomic, we will simply not mention $\Pi R$.

\begin{myeg}[$1$-d Non-relativistic point particle] 
We have that $G=\mathrm{Gal}(0+1,1)$, which corresponds to the Heisenberg group. We label a set of Lie algebra generators of $G$ as $\{T,X,V\}$ with $[V,T]=X$, and all other commutators zero. The other Lie groups involved correspond to $H=\mathbb{R}^2=\{\exp(x X+ v V)\}$, $L=\mathbb{R}=\{\exp(v V)\}$, and $K=\{\id\}$.  The space $A(\mathfrak{g}/\mathfrak{h},\mathfrak{g}/\mathfrak{l})$ has elements given by maps of the form $f_a(cT+\mathfrak{h})=cT+caX+\mathfrak{l}$. The element $\exp(v^\prime V)\in L$ acts on $f_a$ as $f_a\mapsto f_{a+v^\prime}$. The map $\Pi\Omega: f_a\mapsto \exp(a V)\in L/K$ is a valid $L$-map. 
\end{myeg}

\begin{myeg}[$3$-d Non-relativistic point particle]
Now consider the 3-d version of the previous example. We take as $P$ the time translation generator, $C_I$ ($I=1,2,3$) the spatial translations, $B_I$  the boosts, and $J_I$ the rotations, closely following the notation of~\cite{Goon:2012dy}. The symmetry group corresponds to $G=\mathrm{SGal}(0+1,3)=\{e^{tP}e^{\rho^IC_I}e^{v^IB_I}e^{\theta^IJ_I}\}$. The other groups take the form $H=\mathbb{R}^3\rtimes\mathrm{ISO}^+(3)=\{e^{\rho^IC_I}e^{v^IB_I}e^{\theta^IJ_I}\}$, $L= \mathbb{R}^3\rtimes SO(3)=\{e^{v^IB_I}e^{\theta^IJ_I}\}$, and $K=SO(3)=\{e^{\theta^IJ_I}\}$. The affine space $A(\mathfrak{g}/\mathfrak{h},\mathfrak{g}/\mathfrak{l})$ has elements given by $f_{a^I}(cP+\mathfrak{h})=cP+ca^IC_I+\mathfrak{l}$. Under the action of $e^{v^IB_I}e^{\theta^IJ_I}\in L$, $a_I\mapsto \tensor{(e^{\theta^KJ_K})}{_I^J}a_J+v_I$.
The map $\Pi\Omega: f_{a^I}\mapsto e^{a^IV_I}K$ is an $L$-map.
\end{myeg}

\begin{myeg}[$(1+1)$-d, $N=1$ Galileon] 
\emergencystretch 0.7em The symmetry group here is $G=\mathrm{SGal}(1+1,1)$. The group $G$  has Lie algebra generators $\{P_0,P_1,K_1\}$, which generate the $(1+1)$-d Poincar\'e subalgebra and $\{B^0,B^1,C\}$, which have the non-zero commutators $[B^\mu,P_\nu]=\tensor{\eta}{^\mu_\nu}C$, $[K_1,B^0]=-B^1$, and $[K_1,B^1]=-B^0$. We can then write the group $G$ as $\{e^{x^\mu P_\mu}e^{\chi C} e^{\rho_\mu B^\mu} e^{\eta K_1}\}$, the group $H$ as $\{e^{\chi C} e^{\rho_\mu B^\mu} e^{\eta K_1}\}$, the group $L$ as $\{ e^{\rho_\mu B^\mu} e^{\eta K_1}\}$, and the group $K$ as $\{e^{-\eta K_1}\}$. The affine space $A(\mathfrak{g}/\mathfrak{h},\mathfrak{g}/\mathfrak{l})$ has elements given by $f_{a_\mu}(c^\mu P_\mu+\mathfrak{h})=c^\mu P_\mu+c^\mu a_\mu C+\mathfrak{l}$. Under the action of $e^{b_\mu^\prime B^\mu}e^{\eta^\prime K_1}\in L$, denoting $\Lambda(\eta^\prime)=e^{\eta^\prime K_1}$, we get $a_\mu \mapsto \tensor{\Lambda}{_\mu^\nu}(\eta^\prime) a_\nu +b_\mu^{\prime}$. The map $\Pi\Omega:f_{a_\mu}\mapsto e^{a_\mu B^\mu} K$ is an $L$-map.
\end{myeg}

\begin{myeg}[$(3+1)$-d, $N=1$ Galileon]
We now repeat the previous example in the $(3+1)$-d case, so that $G=\mathrm{SGal}(3+1,1)$. The group $G$ has the Lie algebra generators,  $\{P_\mu,K_i,J_i\}$ which generate the $(3+1)$-d  Poincar\'e subalgebra, and $\{B^\mu,P_\nu\}$ which have the non-zero commutators $[P_\nu,B^\mu]=-\tensor{\eta}{_\nu^\mu} C$, $[K_i,B^0]=-B^i$, $[K_i,B^j]=-\tensor{\eta}{^j_i}B^0$, and $[J_i,B^j]=-\epsilon_{ijk}B^k$. We can then write the Lie groups involved as $G=\{e^{x^\mu P_\mu}e^{\chi C}e^{b_\mu B^\mu}e^{\eta^iK_i}e^{\theta^iJ_i}\}$, $H=\{e^{\chi C}e^{b_\mu B^\mu}e^{\eta^iK_i}e^{\theta^iJ_i}\}$, $L=\{e^{b_\mu B^\mu}e^{\eta^iK_i}e^{\theta^iJ_i}\}$, and $K=\{e^{\eta^iK_i}e^{\theta^iJ_i}\}$ which, except for the addition of a rotation, have an identical form to the $(1+1)$-d case. Analogous to what we found above, the affine space $A(\mathfrak{g}/\mathfrak{h},\mathfrak{g}/\mathfrak{l})$ has elements given by $f_{a_\mu}(c^\mu P_\mu +\mathfrak{h})=c^\mu P_\mu+c^\mu a_\mu C+\mathfrak{l}$. Under the action of $e^{b_\mu^\prime B^\mu}e^{\eta^{\prime i}K_i}e^{\theta^{\prime i}J_i} \in L$, denoting $\Lambda^\prime=e^{\eta^{\prime i}K_i}e^{\theta^{\prime i}J_i}$, we get $a_\mu \mapsto \tensor{(\Lambda^\prime)}{_\mu^\nu} a_\nu +b_\mu^{\prime}$, exactly as above. Again, a valid $L$-map is $\Pi\Omega:f_{a_\mu}\mapsto e^{a_\mu B^\mu}K$.
\end{myeg}

\begin{myeg}[$(1+1)$-d Type-1 Superfluid] \label{eg:TypeI_V1}
Here $G$ is the product of the $(1+1)$-d Poincar\'e group and $U(1)$. The group $G$ has the Lie algebra generators  $\{P_0,P_1,K_1\}$ which generate the Poincar\'e subalgebra and $Q$ which generates the subalgebra associated with $U(1)$. We can write $G$ as $G=\{e^{x^\mu P_\mu}e^{\theta Q} e^{\eta K_1}\}$ for $\mu\in \{0,1\}$. The relevant subgroups correspond to $H=\{e^{\theta Q} e^{\eta K_1}\}$, $L=\{e^{\eta K_1}\}$, and $K=\{\id\}$. The space $A(\mathfrak{g}/\mathfrak{h},\mathfrak{g}/\mathfrak{l})$ has elements given by $f_{a_\mu}(c\mu P_\mu+\mathfrak{h})=c^\mu P_\mu+c^\mu a_\mu Q+\mathfrak{l}$. In a similar way to the Galileon example above, under the action of $\Lambda(\eta^\prime)\equiv e^{\eta^\prime K_1}\in L$, we have that $a_\mu\mapsto \tensor{\Lambda}{_\mu^\nu}(\eta^\prime)a_\nu$.

We define the open subset $\Pi R$ as the set of $f_{a_\mu}$ with $a_\mu$ future time-like. A valid choice in $L$-map is then  $ 
\Pi\Omega:f_{a_\mu}\mapsto \exp\left(-\arctanh\frac{a_1}{a_0} K_1\right)$.
\end{myeg}

\begin{myeg}[$(3+1)$-d Type-I Superfluid] 
Turning to the $(3+1)$-d version of the previous example, our group is now the  product of the $(3+1)$-d Poincar\'e group and a $U(1)$. The Lie algebra generators of the $(3+1)$-d Poincar\'e Lie subalgebra, as before, take the form $\{P_\mu,K_i,J_i\}$. The generator of the $U(1)$ Lie algebra is $Q$. We can write the relevant Lie groups as $G=\{e^{x^\mu P_\mu}e^{\phi Q} e^{\eta^iK_i}e^{\theta^iJ_i}\}$, $H=\{e^{\phi Q}e^{\eta^iK_i}e^{\theta^iJ_i}\}$, $L=\{e^{\eta^iK_i}e^{\theta^iJ_i}\}$, and $K=\{e^{\theta^iJ_i}\}$. The affine space $A(\mathfrak{g}/\mathfrak{h},\mathfrak{g}/\mathfrak{l})$ has elements given by $f_{a_\mu}(c^\mu P_\mu+\mathfrak{h})\mapsto c^\mu P_\mu +c^\mu a_\mu Q+\mathfrak{l}$. Under the action of $\Lambda^\prime \in L$, $a_\mu\mapsto \tensor{\Lambda}{^\prime_\mu^\nu}a_\nu$. 

We again need to restrict to an open subset of $A(\mathfrak{g}/\mathfrak{h},\mathfrak{g}/\mathfrak{l})$, $\Pi R$. We define $\Pi R$ by the condition of a future time-like $a_\mu$, the action $\Pi\mathcal{R}$ is that induced by this embedding. The map $\Pi\Omega$ then takes the form $f_{a_\mu}\mapsto \exp\left(\frac{1}{|\vec a|}\arctanh\left(\frac{|\vec a|}{a^0}\right)(a^iK_i)\right)$, where $|\vec a|=\sqrt{a_1^2+a_2^2+a_3^2}$. One can demonstrate the equivariant property of  $\Pi\Omega$ using a slightly technical prescription relying on Thomas-Wigner rotations and related ideas.
\end{myeg}

\begin{myeg}[$(3+1)$-d Solid]
Our last example from the literature corresponds to the $(3+1)$-d Solid. Here $G$ is the  product of the $(3+1)$-d Poincar\'e group and the $3$-d Euclidean group, $\mathrm{ISO}^+(3)$. The generators of the Poincar\'e subalgebra  take the form $\{P_i,K_i,J_i\}$ and those of the $\mathrm{ISO}^+(3)$ subalgebra the form $\{Q_i,\tilde Q_i\}$. Here, $Q_i$ correspond to the translations and $\tilde Q_i$ the rotations. We can write the groups involved as $G=\{e^{x^\mu P_\mu}e^{\rho^iQ_i}e^{\phi^i\tilde Q_i}e^{\eta^i K_i} e^{\theta^iJ_i}\}$, $H=\{ e^{\rho^iQ_i} e^{\phi^i\tilde Q_i} e^{\eta^iK_i} e^{\theta^iJ_i}\}$, $L=\{e^{\phi^i\tilde Q_i} e^{\eta^iK_i} e^{\theta^iJ_i}\}$, and $K=\{e^{\theta^i(J_i+\tilde Q_i)}\}$. The affine space $A(\mathfrak{g}/\mathfrak{h},\mathfrak{g}/\mathfrak{l})$ has elements given by $f_{a^i_\mu}(c^\mu P_\mu+\mathfrak{h})=c^\mu P_\mu+c^\mu a_\mu^i Q_i+\mathfrak{l}$. The group $L$ is the product of the Lorentz group and $SO(3)$ and we can write an element of $L$ as $(\Lambda^\prime,R^\prime)\in L$. The $L$-action on $A(\mathfrak{g}/\mathfrak{h},\mathfrak{g}/\mathfrak{l})$ then takes the form $a^i_\mu\mapsto \tensor{R}{^\prime^j_i}\tensor{\Lambda}{^\prime_\nu^\mu}a^i_\mu$. 

\sloppy Following~\cite{Nicolis:2013lma}, we define $S^\mu=\epsilon^{\mu\alpha\beta\gamma}a^1_\alpha a^2_\beta a^3_\gamma$, and $\tensor{N}{_k^i}=\tensor{(\Lambda_S)}{^\mu_k}a^i_\mu$, where $\Lambda_S=\exp\left(\frac{1}{|\vec S|}\arctanh\left(\frac{|\vec S|}{S^0}\right)(S^iK_i)\right)$, analogous to the above. We define the open subset $\Pi R$ of $A(\mathfrak{g}/\mathfrak{h},\mathfrak{g}/\mathfrak{l})$ by requiring $S^\mu$ to be future time-like and by requiring $\det(N)>0$. The map $\Pi\Omega$ then takes $f_{a^i_\mu}$ to $(\Lambda_S,\sqrt{N^TN} N^{-1})K$. Again, one can show the equivariant property of $\Pi\Omega$ using Thomas-Wigner rotations.
\end{myeg}

\section{Partial actions and constraints} \label{sec:partial_actions}

\subsection{Formalities} 
When the group action is not fibrewise, we need to consider partial actions. Since we will need to consider partial actions in both the topological and smooth contexts we give definitions for both, as follows~\cite{ABADIE200314,10.2307/24718783}.
\begin{mydef}\label{def:Partial_actions}
A partial action of the topological (resp. Lie) group $G$ on the topological space (resp. manifold) $Y$ is a pair $\mathcal{Y}=(\{Y_g\}_{g\in G}, \{\mathcal{Y}_g\}_{g\in G})$ such that:
\begin{enumerate}
\itemsep0em 
\item for all $g\in G$, $Y_g$ are topological spaces (resp. manifolds) that are open topological embeddings (resp. open smooth embeddings), embedded in $Y$ via maps $\mathscr{Y}_g:Y_g\rightarrow Y$, $\mathcal{Y}_g:Y_{g^{-1}}\rightarrow Y_g$ are homeomorphisms (resp. diffeomorphisms) with inverse $\mathcal{Y}_{g^{-1}}=\mathcal{Y}_g^{-1}$, and $Y_e=Y$;
\item the set $\mathcal{U}_Y=\{(g,y)\in G\times Y\mid g\in G, y\in Y_{g^{-1}}\}$ is an open subset of $G\times Y$ and the map $\bar{\mathcal{Y}}:\mathcal{U}_Y\rightarrow Y:(g,y)\mapsto \mathcal{Y}_g(y)$ is continuous (resp. smooth);
\item the action of $\mathcal{Y}_{g_1g_2}$ extends that of $\mathcal{Y}_{g_2}\circ \mathcal{Y}_{g1}$ acting on $(\mathcal{Y}_{g_1})^{-1}(Y_{g_2^{-1}})$.
\end{enumerate}
\end{mydef}
When $Y_{g}=Y$ for all $g\in G$, we return to the usual definition of a continuous (resp. smooth) group action. We used these global actions in \S\ref{sec:homogeneous}.
We next generalise the definitions of the categories ${\tt Fib}_X$ and ${\tt Eta}_X$ in \S\ref{subsec:Categoric_Prelims} to form new categories with a partial action present.
\begin{mydef}
For a Lie group $G$, the category $G\mathen{\tt Fib}_X$ is defined to be the category whose objects are triples $(Y,\pi,\mathcal{Y})$, where $(Y,\pi)$ is a fibred manifold over $X$ and $\mathcal{Y}=(\{Y_g\}_{g\in G},\{\mathcal{Y}_g\}_{g\in G})$ is a partial action of $G$ on $Y$, and whose morphisms between $(Y,\pi,\mathcal{Y})$ and $(Y^\prime,\pi^\prime,\mathcal{Y}^{\prime})$, are fibred morphisms $f:Y\rightarrow Y^\prime$ for which $f(Y_g)\subseteq Y^\prime_g$ and for which the diagram 
\begin{equation} \label{eq:Commutating_Partial_Morphism}
\begin{tikzcd}
Y_{g^{-1}}\arrow{r}{f} \arrow[swap]{d}{\mathcal{Y}_g}& Y_{g^{-1}}^\prime\arrow{d}{\mathcal{Y}_g^\prime}\\
Y_g \arrow[swap]{r}{f} & Y_g^\prime
\end{tikzcd}
\end{equation}
commutes, for all $g \in G$.
\end{mydef}
\begin{mydef}
\sloppy For a topological group $G$, the category $G\mathen{\tt Eta}_X$ is defined to be the category whose objects are triples $(E,p,\mathcal{E})$ where $(E,p)$ is an \'etal\'e space over $X$ and $\mathcal{E}=(\{E_g\}_{g\in G}, \{\mathcal{E}_g\}_{g\in G})$ is a partial action of $G$ on $E$, and whose morphisms between $(E,p,\mathcal{E})$ and $(E^\prime,p^\prime,\mathcal{E}^{\prime})$ are \'etal\'e morphisms $f:E\rightarrow E^\prime$ which satisfy $f(E_g)\subseteq E^\prime_g$ and the analogous commutative diagram to~\ref{eq:Commutating_Partial_Morphism}.
\end{mydef}

We now let $G$ be a Lie group, corresponding to the symmetry group of our system. The corresponding category of \'etal\'e spaces is $G^{d}\mathen{\tt Eta}_X$, where the topological group $G^d$ is the group $G$ equipped with the discrete topology. Our functors $\Gamma$ and $J^r$ can then be modified to account for partial actions as follows.\footnote{The proof that these functors are well defined is given in Appendix~\ref{ap:partial_actions_for_functors}.}
\begin{mydef}\label{def:Local_Sections_Functor_Partial_Action}
The {\em equivariant local sections functor} $\Gamma:G\mathen{\tt Fib}_X\rightarrow G^{d}\mathen{\tt Eta}_X$ takes $(Y,\pi,\mathcal{Y})$ to $(\Gamma Y,\Gamma\pi,\Gamma \mathcal{Y})$, with $\Gamma\mathcal{Y}:=(\{\Gamma Y_g\}_{g\in G^d},\{\Gamma \mathcal{Y}_g\}_{g\in G^d})$ and
\begin{align}
\Gamma Y_g&:=\{[\beta]_x\in \Gamma Y\mid\beta(x)\in Y_{g}\ \forall x\in \dom(\beta), \pi\circ Y_{g^{-1}}\circ \beta\text{ is an open embedding}\},\\
\Gamma\mathcal{Y}_{g}&:\Gamma Y_{g^{-1}}\rightarrow \Gamma Y_g:[\beta]_g\mapsto [Y_{g}\circ \beta\circ h_{g,\beta}^{-1}]_{h_{g,\beta}(x)},
\end{align}
where $h_{g,\beta}$ is the map defined by $\pi\circ Y_{g}\circ \beta$, but with its codomain restricted to be its image. The functor
 $\Gamma$ takes the morphism $f:Y\rightarrow Y^\prime$ to $f:\Gamma Y\rightarrow \Gamma Y^\prime:[\alpha]_x\mapsto [f\circ \alpha]_x$.
\end{mydef}
\begin{mydef} \label{def:Jet_Functor_Partial_Action}
The {\em equivariant $r$th-jet functor} $J^r:G\mathen{\tt Fib}_X\rightarrow G\mathen{\tt Fib}_X$ takes $(Y,\pi,\mathcal{Y})$ to $(J^rY,\pi^r,J^r\mathcal{Y})$, with $J^r\mathcal{Y}:=(\{J^rY_g\}_{g\in G}, \{J^r\mathcal{Y}_g\}_{g\in G})$ and
\begin{align}
J^rY_g&=\{j^r_x\beta \in J^rY \mid \beta(x)\in Y_{g}\ \forall x\in \dom(\beta), \pi\circ Y_{g^{-1}}\circ \beta\text{ is an open embedding}\},\\
J^r\mathcal{Y}_{g}&:J^rY_{g^{-1}}\rightarrow J^rY_g:j^r_x\beta \mapsto j^r_{h_{g,\beta}(x)}(Y_{g}\circ \beta\circ h_{g,\beta}^{-1}),
\end{align}
where again, $h_{g,\beta}$ is the map defined by $\pi\circ Y_{g}\circ \beta$, but with its codomain restricted to its image. The functor $J^r$ takes the morphism $f:Y\rightarrow Y^\prime$ to $J^r f:J^r Y\rightarrow J^rY^\prime:j^r_x\alpha\mapsto j^r_x(f\circ \alpha)$.
\end{mydef}
The functors $\Gamma$ and $J^r$ preserve the same properties listed in Lemmas \ref{th:Things_preserved_by_gamma} and \ref{th:Structures_preserved_by_jet_functor}. In addition, we have the following
\begin{mylem} \label{th:open_embeddings_preserved}
Say a morphism  between $(Z,\zeta,\mathcal{Z})$ and $(Y,\pi,\mathcal{Y})$ in $G\mathen{\tt Fib}_X$ is an \emph{embedding of partial actions} if the underlying fibred morphism $\iota:Z\rightarrow Y$ is an embedding such that $Z_g =\iota^{-1}(Y_g)$ for all $g \in G$, along with the analogous statement for $G\mathen{\tt Eta}_X$. The functors $J^r$ and $\Gamma$ preserve embeddings of partial actions. (Proof: Appendix~\ref{ap:Functors_preserve_embeddings_PA})
\end{mylem}

Returning to natural transformations, we have the following
\begin{myprop} \label{th:natural_transformations_pa}
The maps $\pi^{r,l}:J^rY\rightarrow J^lY$ form a natural transformation of functors $G\mathen{\tt Fib}_X \to G\mathen{\tt Fib}_X$. The maps $j^r:\Gamma Y\rightarrow \Gamma J^rY$ form a natural transformation in of functors $G\mathen{\tt Fib}_X \to G^{d}\mathen{\tt Eta}_X$. (Proof: Appendix~\ref{ap:natural_transformations_pa})
\end{myprop}

We now go on to constraints, which we express by a single
\begin{myth} \label{th:results_apply_to_partial_actions}
The results given in \S\ref{sec:Constraints} hold with the categories  ${\tt Fib}_X$ and ${\tt Eta}_X$ replaced with $G\mathen{\tt Fib}_X$ and $G^{d}\mathen{\tt Eta}_X$, with the  functors replaced by their corresponding equivariant versions defined in \ref{def:Local_Sections_Functor_Partial_Action} and \ref{def:Jet_Functor_Partial_Action},  and with `embeddings' replaced with `embeddings of partial actions'. (Proof: Appendix~\ref{app:Proof_S3_5})
\end{myth}
\subsection{Examples}
We now examine three physical examples using the framework of partial actions.
\begin{myeg}[$(1+1)$-d Type-I superfluid] 
As a warm up, we re-examine the Type-I superfluid in Example~\ref{eg:TypeI_V1}. Here all our group actions will in fact be global. The symmetry group, $G$, corresponds to the product of $U(1)$ and the $(1+1)$-d Poincar\'e group. A general element of this group will be specified by $(x^{\prime \mu},\eta^\prime,\theta^\prime)\in \mathbb{R}^2\times S^1$. With $X=\mathbb{R}^2$, we first specify the fibred quotient in $G\mathen{\tt Fib}_X$ given by the fibred manifolds
\begin{align}
(Y,\pi,\mathcal{Y})&=(\mathbb{R}^3\times S^1,(x^\mu,\eta,\theta)\mapsto x^\mu,(\{Y\}_{g\in G},\nonumber \\& \{\mathcal{Y}_g:(x^\mu,\eta,\theta)\mapsto (x^{\prime\mu}+\tensor{\Lambda}{^\mu_\nu}(\eta^\prime)x^\nu,\eta+\eta^\prime,\theta+\theta^\prime)\}_{g\in G})),\nonumber \\
(Z,\zeta,\mathcal{Z})&=(\mathbb{R}^2\times S^1,(x^\mu,\theta)\mapsto x^\mu,(\{Z\}_{g\in G}, \nonumber \\ & \{\mathcal{Z}_g:(x^\mu,\theta)\mapsto (x^{\prime\mu}+\tensor{\Lambda}{^\mu_\nu}(\eta^\prime)x^\nu,\theta+\theta^\prime)\}_{g\in G})).
\end{align}
and the surjective submersion $\tau_Z:Y\rightarrow Z:(x^\mu,\eta,\theta)\mapsto (x^\mu,\theta)$. The first jet manifolds of $Y$ and $Z$ are
\begin{align}
(J^1Y,\pi^1,J^1\mathcal{Y})&=(\mathbb{R}^3\times S^1\times \mathbb{R}^4,(x^\mu,\eta,\theta,\eta_\mu,\theta_\mu)\mapsto x^\mu,\nonumber \\
& (\{J^1Y\}_{g\in G},\{J^1\mathcal{Y}_g:(\cdots,\eta_\mu,\theta_\mu)\mapsto(\cdots,\tensor{\Lambda}{_\mu^\nu}(\eta^\prime)\eta_\nu,\tensor{\Lambda}{_\mu^\nu}(\eta^\prime)\theta_\nu)\}_{g\in G})), \nonumber\\ 
(J^1Z,\zeta^1,J^1\mathcal{Z})&=(\mathbb{R}^2\times S^1\times \mathbb{R}^2,(x^\mu,\theta,\theta_\mu)\mapsto x^\mu,\nonumber \\
& (\{J^1Z\}_{g\in G},\{J^1\mathcal{Z}_g:(\cdots,\theta_\mu)\mapsto(\cdots,\tensor{\Lambda}{_\mu^\nu}(\eta^\prime)\theta_\nu)\}_{g\in G})).
\end{align}
We want to form a comeronomic constraint, so introduce another fibred manifold
\begin{align}
(R,\rho,\mathcal{R})&=(\mathbb{R}^2\times S^1\times \mathbb{R}^2,(x^\mu,\theta,z_1,z_2)\mapsto x^\mu,\nonumber \\
&(\{R\}_{g\in G},\{\mathcal{R}_g:(x^\mu,\theta,z_1,z_2)\mapsto (\tensor{\Lambda}{^\mu_\nu}(\eta^\prime)x^\nu+x^{\prime \mu},\theta+\theta^\prime, z_1+\eta^\prime,z_2)\}_{g\in G})),
\end{align}
with the open embedding of partial actions 
\begin{align}
\iota_R:R\rightarrow J^1Z:(x^\mu,\theta,z_1,z_2)\mapsto  (x^\mu,\theta,\cosh z_1 e^{z_2},-\sinh z_1 e^{z_2}).
\end{align} Physically $R$ manifests the condition of restricting to $\theta_\mu$ (as introduced in the definition of $J^1Z$) that are future time-like vectors. We then can choose $\Omega:R\rightarrow Y:(x^\mu,\theta,z_1,z_2)\mapsto (x^\mu,z_1,\theta)$. This choice gives 
\begin{align}
(Q,\nu,\mathcal{Q})=(\mathbb{R}^2\times S^1\times \mathbb{R}^4,(x^\mu,\theta,z_1,z_2,\eta_\mu)\mapsto x^\mu,(\{Q\}_{g\in G},\{\mathcal{Q}_g\}_{g\in G})),
\end{align}
where $\mathcal{Q}_g$ is as suggested by the notation. A point in $E^Q$ considered as a subspace of $\Gamma Z$  is of the form $[x^\mu\mapsto (x^\mu,-\arctanh(\partial_1\theta(x^\mu)/\partial_0\theta(x^\mu)),\theta(x^\mu)))]_{x^\mu}$, for future time-like $\partial_\mu \theta$. Under the isomophrism of \'etal\'e spaces this gets mapped to the point $[x^\mu\mapsto (x^\mu,\theta(x^\mu))]_{x^\mu}$ in $ E^R$, considered as a subspace of $\Gamma Z$.
\end{myeg}

\begin{myeg}[$(1+1)$-d relativistic particle] 
Here the symmetry group $G$ is the $(1+1)$-d Poincar\'e group, an element of which we specify by $(x^{\prime \mu},\eta^\prime)\in \mathbb{R}^3$, for $\mu\in \{0,1\}$. We have $X=\mathbb{R}$, and the fibred quotient in $G\mathen{\tt Fib}_X$ defined by
\begin{align}
(Y,\pi,\mathcal{Y})&=(\mathbb{R}^3,(x^\mu,\eta)\mapsto x^0,(\{Y\}_{g\in G}, \{\mathcal{Y}_g:(x^\mu,\eta)\mapsto (x^{\prime\mu}+\tensor{\Lambda}{^\mu_\nu}(\eta^\prime)x^\nu,\eta+\eta^\prime)\}_{g\in G})),\nonumber \\
(Z,\zeta,\mathcal{Z})&=(\mathbb{R}^2,(x^\mu)\mapsto x^0,(\{Z\}_{g\in G}, \{\mathcal{Z}_g:x^\mu\mapsto x^{\prime\mu}+\tensor{\Lambda}{^\mu_\nu}(\eta^\prime)x^\nu\}_{g\in G})).
\end{align}
and the surjective submersion $\tau_Z:Y\rightarrow Z:(x^\mu,\eta)\mapsto x^\mu$. The corresponding first jet manifolds are
\begin{align}
&(J^1Y,\pi^1,J^1\mathcal{Y})=\Bigg(\mathbb{R}^{5},(x^\mu,\eta,x^1_0,\eta_0)\mapsto x^0,
\Bigg(\{ \{(x^\mu,\eta,x^1_0,\eta_0)\mid \tensor{\Lambda}{^0_0}(\eta^\prime)\ne \tensor{\Lambda}{^0_1}(\eta^\prime)x^1_0\}\}_{g\in G}, \nonumber \\& 
\left\{J^1\mathcal{Y}_g:(\cdots,x^1_0,\eta_0)\mapsto \left(\cdots, \frac{\tensor{\Lambda}{^1_0}(\eta^\prime)+\tensor{\Lambda}{^1_1}(\eta^\prime)x^1_0}{\tensor{\Lambda}{^0_0}(\eta^\prime)+\tensor{\Lambda}{^0_1}(\eta^\prime)x^1_0},\frac{\eta_0}{\tensor{\Lambda}{^0_0}(\eta^\prime)+\tensor{\Lambda}{^0_1}(\eta^\prime)x^1_0}\right)\right\}_{g\in G}\Bigg)\Bigg),
\end{align}
and
\begin{align}
(J^1Z,\zeta^1,J^1\mathcal{Z})=\Bigg(\mathbb{R}^{3},(x^\mu,\eta,x^1_0)\mapsto x^0,
\Bigg(\{ \{(x^\mu,x^1_0)\mid \tensor{\Lambda}{^0_0}(\eta^\prime)\ne \tensor{\Lambda}{^0_1}(\eta^\prime)x^1_0\}\}_{g\in G}, \nonumber \\ 
\left\{J^1\mathcal{Z}_g:(\cdots,x^1_0)\mapsto \left(\cdots, \frac{\tensor{\Lambda}{^1_0}(\eta^\prime)+\tensor{\Lambda}{^1_1}(\eta^\prime)x^1_0}{\tensor{\Lambda}{^0_0}(\eta^\prime)+\tensor{\Lambda}{^0_1}(\eta^\prime)x^1_0}\right)\right\}_{g\in G}\Bigg)\Bigg).
\end{align}
Notice that although the group action on $Y$ and $Z$ is the global one, the action on $J^1Y$ and $J^1Z$ is strictly partial. We then choose $\Omega:J^1Z\rightarrow Y:(x^\mu,x^1_0)\mapsto (x^\mu,\arctanh x_0^1)$, which leads to the coholonomic constraint
\begin{multline}
(Q,\nu,\mathcal{Q})=(\mathbb{R}^5,(x^\mu,x^1_0,\eta_0)\mapsto x^0,\\
(\{\{(x^0,x_0^1,\eta_0)\mid\tensor{\Lambda}{^0_0}(\eta^\prime)\ne\tensor{\Lambda}{^0_1}(\eta^\prime)x^1_0\} \}_{g\in G},\{\mathcal{Q}_g\}_{g\in G}))
\end{multline}
where $\mathcal{Q}_g$ is as suggested by notation. The manifold $Q$ is embedded into $J^1Y$ via $\iota_Q:(x^\mu,x^1_0,\eta_0)\mapsto (x^\mu,\arctanh x^1_0, x^1_0, \eta_0)$. An element of $E^Q$, as a subspace of $\Gamma Y$, is  $[x^0\mapsto (x^0, x^1(x^0),\arctanh \partial_0x^1(x^0))]_{x^0}$, under the isomorphism of \'etal\'e spaces, this gets mapped into $[x^0\mapsto (x^0,x^1(x^0))]_{x^0}$ in $\Gamma Z$.
\end{myeg}
\begin{myeg}[String in a plane]
This follows the same pattern as the $(1+1)$-d relativistic  particle discussed above. It was previously studied in~\cite{Low:2001bw}, and in fact the mathematical set-up coincides with a system studied as an example in~\cite{olver2000applications} in the context of symmetries of differential equations.  Here $G$ is the Euclidean group in $2$-d, a general element of which we label by $(x^\prime, y^\prime,\theta^\prime)\in \mathbb{R}^2\times S^1$. We have $X=\mathbb{R}$, and the fibred quotient in $G\mathen{\tt Fib}_X$ given by 
\begin{align}
(Y,\pi,\mathcal{Y})&=(\mathbb{R}^2\times S^1,(x,y,\theta)\mapsto x,(\{Y\}_{g\in G}, \nonumber \\ &\{\mathcal{Y}_g:(x,y,\theta)
\mapsto (x^\prime+x\cos\theta^\prime+y\sin\theta^\prime,y^\prime+y\cos\theta^\prime-x\sin\theta^\prime,\theta+\theta^\prime)\}_{g\in G})),\nonumber \\
(Z,\zeta,\mathcal{Z})&=(\mathbb{R}^2\times S^1,(x,y)\mapsto x,(\{Z\}_{g\in G}, \nonumber \\ &\{\mathcal{Z}_g:(x,y)
\mapsto (x^\prime+x\cos\theta^\prime+y\sin\theta^\prime,y^\prime+y\cos\theta^\prime-x\sin\theta^\prime)\}_{g\in G})).
\end{align}
and the surjective submersion  $\tau_Z:Y\rightarrow Z:(x,y,\theta)\mapsto (x,y)$. The corresponding first jet manifolds are
\begin{multline}
(J^1Y,\pi^1,J^1\mathcal{Y})=\Bigg(\mathbb{R}^{2}\times S^1\times \mathbb{R}^2,(x,y,\theta,y_x,
\theta_x)\mapsto x,
\\ \Bigg(\{ \{(x,y,\theta,y_x,
\theta_x)\mid \cos\theta^\prime\ne- y_x\sin\theta^\prime\}\}_{g\in G}, \nonumber \\ 
\left\{J^1\mathcal{Y}_g:(\cdots,y_x,\theta_x)\mapsto \left(\cdots, \frac{\sin\theta^\prime+y_x\cos\theta^\prime}{\cos\theta^\prime-y_x\sin\theta^\prime},\frac{\theta_x}{\cos\theta^\prime-y_x\sin\theta^\prime}\right)\right\}_{g\in G}\Bigg)\Bigg),
\end{multline}
and
\begin{align}
(J^1Z,\zeta^1,J^1\mathcal{Z})=\Bigg(\mathbb{R}^{3},(x,y,y_x)\mapsto x,
\Bigg(\{ \{(x,y,y_x)\mid \cos\theta^\prime\ne-y_x\sin\theta^\prime\}\}_{g\in G}, \nonumber \\ 
\left\{J^1\mathcal{Z}_g:(\cdots,y_x)\mapsto \left(\cdots, \frac{\sin\theta^\prime+y_x\cos\theta^\prime}{\cos\theta^\prime-y_x\sin\theta^\prime}\right)\right\}_{g\in G}\Bigg)\Bigg).
\end{align}
Forming the coholonomic constraint associated with $\Omega:(x,y,y_x)\mapsto(x,y,\arctan y_x)$, we get
\begin{multline}
(Q,\nu,\mathcal{Q})=(\mathbb{R}^4,(x,y,y_x,\theta_x)\mapsto x^0,\\( \{\{(x,y,y_x,\theta_x)\mid \cos\theta^\prime\ne- y_x\sin\theta^\prime\} \}_{g\in G},\{\mathcal{Q}_g\}_{g\in G}\})),
\end{multline}
where, as above, $\mathcal{Q}_g$ is as suggested by the notation. The embedding of $Q$ into $J^1Y$ is given by $\iota_Q:(x,y,y_x,\theta_x)\mapsto(x,y,\arctan y_x,y_x,\theta_x)$. A typical element of $E^Q$, as a subspace of $\Gamma Y$, is of the form $[x\mapsto (x,y(x),\arctan \partial_xy(x))]_x$, which maps to $[x\mapsto (x,y(x))]_x$ in $\Gamma Z$ under the isomorphism of \'etal\'e spaces.
\end{myeg}

\section*{Acknowledgments}
This work is supported by STFC consolidated grants ST/P000681/1 and ST/S505316/1. 

\appendix

\section{Proofs for \S\ref{sec:mathematical_prerequisites}} \label{app:Proof_S2}
Throughout this Appendix we let $(Y,\pi)$ and $(Z,\zeta)$ be fibred manifolds in  ${\tt Fib}_X$, and $f$ a morphism in this category.
\subsubsection*{Proof of Lemma~\ref{th:Things_preserved_by_gamma}} \label{ap:Things_preserved_by_gamma}
{\em Injections: } Let $f:Z\rightarrow Y$ be an injection. Recalling that $\Gamma f[\alpha]_x=\Gamma f[\beta]_x\Leftrightarrow[f\circ \alpha]_x=[f\circ \beta]_x$, then there is a $\tilde \alpha\in [\alpha]_x $ and a $\tilde \beta \in [\beta]_x$ such that $f\circ \tilde \alpha=f\circ \tilde \beta$, but since $f$ is an injection this implies $\tilde \alpha=\tilde \beta$ so $[\alpha]_x=[\beta]_x$. Thus if $f$ is an injection, so is $\Gamma f$. But $\Gamma f$ is an open map, as can be seen from its explicit form, and the topology on the \'etal\'e spaces, and an injective open map is a topological embedding, thus $\Gamma f$ is a topological embedding.
\newline
~\newline
\noindent
{\em Counterexample for surjections: } Let $X=\mathbb{R}$ and $(Y,\pi)=(\mathbb{R}^2,(x,y)\mapsto x)$, and let $f:Y\rightarrow Y: (x,y)\mapsto (x,y^3)$, which is a surjection. Let $\alpha:x\mapsto (x,x)$, then $[\alpha]_{x=0}$ is not in the image of $\Gamma f$, since $x^{1/3}$ is not smooth at the origin.
\subsubsection*{Proof of  Lemma~\ref{th:Structures_preserved_by_jet_functor}} \label{ap:Structures_preserved_by_jet_functor}

{\em Submersions: } Let $f:Y\rightarrow Z$ be a submersion. Then around every $y\in Y$ there is a neighbourhood $U_y$ which has coordinates $(x^\mu,z^i,y^a)$ such that $f(U_y)$ (a submersion is open) has coordinates $(x^\mu,z^i)$ with $f:(x^\mu,z^i,y^a)\mapsto (x^\mu,z^i)$. The open subset $(\pi^{r,0})^{-1}(U_y)$ then has coordinates $(x_\mu,z^i,y^a,z^i_I,y^a_I)$, for multi-indices $I$, whist $(\zeta^{r,0})^{-1}(f(U_y))$ has coordinates $(x_\mu,z^i,z^i_I)$,  where $J^rf:(x_\mu,z^i,y^a,z^i_I,y^a_I)\rightarrow (x^\mu,z^i,z^i_I)$. This map is clearly also a submersion.
\newline
~\newline
\noindent
{\em Surjective submersions: } Let $f:Y\rightarrow Z$ be a surjective submersion. This follows directly from the case of a submersion, by noting that every point $z\in Z$ sits in a neighbourhood $f(U_y)$ as constructed above. 
\newline
~\newline
\noindent
{\em Embeddings: } Let $f:Z\rightarrow Y$ be an embedding. For every $f(z)$ there is a neighbourhood of $Y$, $U_{f(z)}$, which has coordinates $(x^\mu,z^a,y^i)$, such that $f^{-1}(U_{F(z)})$ has coordinates $(x^\mu,z^a)$ with $f: (x^\mu,z^a)\mapsto (x^\mu,z^a,0)$.  In the corresponding induced coordinates, $J^rf: (x^\mu,z^a,z^a_I) \mapsto (x^\mu, z^a, 0, z^a_I, 0)$. Since $J^rf$ in these coordinates maps an open subset to an open subset in the induced topology of its image, it is manifestly an immersion in these coordinates, and since it injectively maps the fibre above $(x^\mu, z^a)$ to the fibre above $(x^\mu,z^a,0)$, it is an embedding.
\newline
~\newline
\noindent
{\em Immersions: } An immersion is equivalent to a local embedding, and thus this follows from the above.
\newline
~\newline
\noindent
{\em Injective immersions: } Let $f:Z\rightarrow Y$ be an injective immersion. For every $z\in Z$ there is a neighbourhood $V_z$, such that there is a neighbourhood around $f(z)$, $U_{f(z)}$ and coordinates on these neighbourhoods with $f:(x_\mu,z^i)\mapsto (x_\mu,z^i,y^a=0)$. Using the induced coordinates on $(\zeta^{r,0})^{-1}(V_z)$ and $(\pi^{r,0})^{-1}(U_{F(z)})$, we get $J^rf:(x_\mu,z^i,z^i_I)\mapsto (x_\mu,z^i,0,z^i_I,0)$. Each of these coordinates covers its respective fibers of {\em e.g.} $J^rZ\rightarrow Z$ and, since the map between $Z$ and $Y$ is an injection, we can see that $J^rf$ is an injection. The form of $J^rf$ in these local coordinates also indicates that it is an immersion. 
\newline
~\newline
\noindent
{\em Counterexamples for injections and surjections: } Let $X=\mathbb{R}$ and $(Y,\pi)=(\mathbb{R}^2,(x,y)\mapsto x)$, and let $f:Y\rightarrow Y:(x,y)\mapsto (x,y^3)$, which is a bijection. We have $J^1f: (x,y,y_x)\mapsto (x,y^3,3 y_x y^2)$, which is neither surjective (since \emph{e.g.} $(x,0,1)$ is not in the image) nor injective (since \emph{e.g.} $J^1f((x,0,1))=J^1f((x,0,2))$).
\section{Proofs for \S\ref{sec:partial_actions}} \label{app:Proof_S5}
Throughout this Appendix, $(Y,\pi,\mathcal{Y})$ is an object in $G\mathen{\tt Fib}_X$. Further, the statement that, {\em e.g.}, $[\beta]_x\in \Gamma Y_g$ will be understood to imply that we are taking $\beta$ to have a small enough domain that it satisfies the conditions in the definition of $\Gamma Y_g$, and similarly for $j^r_x\beta$.
\subsubsection*{Proof that \ref{def:Local_Sections_Functor_Partial_Action} and \ref{def:Jet_Functor_Partial_Action} are well defined} \label{ap:partial_actions_for_functors}
We must show that the partial actions in the definitions of $\Gamma$ and $J^r$ are indeed partial actions and $\Gamma$ and $J^r$ yield {\em bona fide} morphisms in the codomain. We deal with them in turn. 

{\em The partial action $\Gamma\mathcal{Y}$:}
We check that the list of properties in Def.~\ref{def:Partial_actions} hold for $\Gamma \mathcal{Y}$.
\begin{enumerate}
\item From their definition, $\Gamma Y_g$ are open in $\Gamma Y$. The maps $\Gamma\mathcal{Y}_g$ are manifestly open maps. 

We now want to show that $\Gamma \mathcal{Y}_g$ and $\Gamma\mathcal{Y}_{g^{-1}}$ are mutually inverse. Since both $\Gamma\mathcal{Y}_g$ and $\Gamma\mathcal{Y}_{g^{-1}}$ are open, this will also show not only that they are both continuous, but also that the image of $\Gamma\mathcal{Y}_g$, say, really is $\Gamma Y_g$. Let $[\beta]_x\in \Gamma Y_{g^{-1}}$, and $\beta^\prime=\mathcal{Y}_g\circ \beta\circ h_{g,\beta}^{-1}$. Then
\begin{align}
\pi\circ  \mathcal{Y}_{g^{-1}}\circ  \beta^\prime=\pi \circ  \mathcal{Y}_{g^{-1}}\circ \mathcal{Y}_g\circ  \beta \circ h_{g,\beta}^{-1}=\pi\circ \beta \circ h_{g,\beta}^{-1}=h_{g,\beta}^{-1}.
\end{align}
Hence, $\pi\circ  \mathcal{Y}_{g^{-1}}\circ  \beta^\prime$ is an open embedding and have $h_{g^{-1},\beta^\prime}=h_{g,\beta}^{-1}$. Acting on $[\beta]_x\in \Gamma Y_{g^{-1}}$ with $\Gamma\mathcal{Y}_{g^{-1}}\circ \Gamma\mathcal{Y}_{g}$ we get
\begin{equation}
\Gamma\mathcal{Y}_{g^{-1}}\circ \Gamma\mathcal{Y}_{g}([\beta]_x)=[\mathcal{Y}_{g^{-1}}\circ \mathcal{Y}_g\circ \beta \circ h_{g,\beta}^{-1}\circ h_{g,\beta}]_{h_{g,\beta}^{-1}\circ h_{g,\beta}(x)}=[\beta]_x,
\end{equation}
so $\Gamma\mathcal{Y}_{g^{-1}}$ and $ \Gamma\mathcal{Y}_{g}$ are indeed mutually inverse. 

Next we turn our attention to the case when $g=e$ (the identity of $G$). Looking at the definition of $\Gamma Y_g$, for $g=e$, for any local section $\beta(x)\in Y_e=Y$, and $\pi\circ \mathcal{Y}_e\circ\beta=\iota_{U,X}$, which is an open embedding, thus $\Gamma Y_e=\Gamma Y$. For each $\beta \in \mathcal{Y}_e$, $h_{e,\beta}=\id_U$, and thus, from its definition, $\Gamma\mathcal{Y}_e$ is indeed the identity on $\Gamma Y$.

\item The condition that $\Gamma\mathcal{U}_Y$ is open, and that $\Gamma\bar{\mathcal{Y}}$ is continuous, follows trivially from the fact we chose the discrete topology on $G$ (recall that $\Gamma:G\mathen {\tt Fib}_X\rightarrow G^{d}\mathen{\tt Eta}_X$). If we had not done so, then generically $\Gamma\bar{\mathcal{Y}}$ would not be continuous. 

\item  For an arbitrary point, $[\beta]_x\in (\Gamma \mathcal{Y}_{g_1})^{-1}(\Gamma Y_{g_2^{-1}})$, we can take $\beta:U\rightarrow Y$ such that $\beta(x)\in  ( \mathcal{Y}_{g_1})^{-1}( Y_{g_2^{-1}})$. For such a $\beta$, we have that $\pi\circ \mathcal{Y}_{g_1g_2}\circ \beta=\pi\circ \mathcal{Y}_{g_1}\circ\mathcal{Y}_{g_2}\circ \beta$, the left hand side of which, given the form of $\beta$, must be an open embedding, and therefore the right hand side must be too. This tells us that $ (\Gamma \mathcal{Y}_{g_1})^{-1}(\Gamma Y_{g_2^{-1}})\subseteq \Gamma Y_{(g_1g_2)^{-1}}$. From their explict actions, it can then be seen that the action of $\Gamma \mathcal{Y}_{g_1g_2}$ extends the action of $\Gamma \mathcal{Y}_{g_1}$ followed by $\Gamma \mathcal{Y}_{g_2}$.
\end{enumerate}

{\em The partial action $J^r\mathcal{Y}$:} As for $\Gamma\mathcal{Y}$ we follow the list given in Def.~\ref{def:Partial_actions}, but now the proof is somewhat more involved, since we must check smoothness in addition.
\begin{enumerate}
\item Firstly we need to show that $J^r Y_g$ are open. Let $(x^\mu,y^a, y^a_\mu)$ be some induced coordinates of $U\subseteq J^1Y$, so that $\pi^{1,0}(U)\in Y_g$. Let $(x^{\prime \mu},y^{\prime a})$ be some adapted coordinates of the image of $\pi^{1,0}(U)$ under $\mathcal{Y}_g$ (which can be made to exist by making $U$ small enough). Then the condition on whether a point $(x^\mu,y^a,y^a_\mu)$ is in $(J^rY_g)$ can be expressed in terms of the Jacobian function defined locally on $U$ by
\begin{align}
\mathrm{jac}:p\rightarrow \det(D_\mu(x^{\mu\prime}\circ \mathcal{Y}_g\circ \pi^{1,0})|_{(x^\mu,y^a,y^a_\mu)}) \text{ where } D_\mu=\frac{\partial}{\partial x^\mu}+y^a_\mu \frac{\partial}{\partial y^a}.
\end{align}
The points in $U\cap J^1Y_g$ correspond to those in $\mathrm{jac}^{-1}(\mathbb{R}-\{0\})$ which is  open, since $\mathbb{R}-\{0\}$ is open and $\mathrm{jac}$ is continuous. The union of all such open subsets for all $U$ is $J^1Y_g$, which is therefore open in $J^1Y$. For generic $r$, $J^rY_g=(\pi^{r,0})^{-1}(J^1Y_g)$  are open for all $g \in G$ since $\pi^{r,0}$ is continuous.

We now need to show the smoothness of $J^r\mathcal{Y}_g$. In our local coordinates above, for $(x^\mu,y^a,y^a_\mu)\in U\cap J^1Y_g$ we define the matrix
\begin{align}
\tensor{M}{_\nu^\mu}=D_\nu(x^{\prime \mu} \circ \mathcal{Y}_g\circ \pi^{1,0})|_{(x^\mu,y^a,y^a_\mu)},
\end{align}
which is essentially the Jacobian matrix, which given our definition of $J^r\mathcal{Y}_g$ is invertible on this space. To determine the smoothness of $J^1 \mathcal{Y}_g$ we can look at its value in the induced coordinates associated with $(x^{\prime \mu},y^{\prime a})$ on $J^1Y$, $(x^{\prime \mu},y^{\prime a},y^{\prime a}_{\mu})$. The smoothness in the coordinates $x^{\prime\mu}$ and $y^{\prime a}$ follows directly from that of $\mathcal{Y}_g$. For $y^{\prime a}_\mu$ we have   
\begin{align}
y_\mu^{\prime a}\circ J^1 \mathcal{Y}_g(x^\mu,y^a,y^a_\mu)=\tensor{(M^{-1})}{_\nu^\mu}D_\mu(y^{\prime a}\circ \mathcal{Y}_g \circ \pi^{1,0})|_{(x^\mu,y^a,y^a_\mu)},
\end{align}
which is indeed smooth.  For generic $r$, the smoothness of $J^r\mathcal{Y}_g$ follows from the smoothness of $J^1\cdots J^1\mathcal{Y}_g$, noting that $J^rY$ is embedded, via an embedding of partial actions, into $J^1\cdots J^1Y$ (for $r$, $J^1$'s) and one can pick out the appropriate coordinates to show smoothness.

The property that $J^r\mathcal{Y}_g$ and $J^r \mathcal{Y}_{g^{-1}}$ are mutually inverse follows in the same way as for $\Gamma \mathcal{Y}_g$ and $\Gamma \mathcal{Y}_{g^{-1}}$.
\item The argument that $J^r\mathcal{U}_Y$ is open and $J^r\bar{\mathcal{Y}}$ is smooth follows in exactly the same way as our arguments showing that $J^rY_g$ are open and that $J^r\mathcal{Y}_g$ is smooth. In effect it follows from the smoothness of $\bar{\mathcal{Y}}$ and the (locally defined) Jacobian.
\item The property that $J^r\mathcal{Y}_{g_1g_2}$ extends the combined action of $J^r \mathcal{Y}_{g_1}$ and $J^r \mathcal{Y}_{g_2}$, follows in the same way as for $\Gamma$.
\end{enumerate}

{\em Target morphisms of $\Gamma$ and $J^r$:}
We want to show, for $f$ a morphism in $G\mathen{\tt Fib}_X$, that $\Gamma f$ is a morphism in $G^{d}\mathen{\tt Eta}_X$ and that $J^r f$ is a morphism in $G\mathen{\tt Fib}_X$. We show it for $\Gamma f$, noting that for $J^r f$ the proof works analogously. Let $(Y,\pi,\mathcal{Y})$ and $(Y^\prime,\pi^\prime,\mathcal{Y}^\prime)$ be two objects in $G\mathen{\tt Fib}_X$ and $f:Y\rightarrow Y^\prime$ be a morphism between these two objects. We need to show that $\Gamma f$ interacts with our partial actions correctly (the other required properties hold trivially). Thus, let $[\beta]_x\in \Gamma Y_g$. We need that $\Gamma f[\beta]_x=[f\circ \beta]_x\in \Gamma Y^\prime_g$.  Since $f$ is a morphism in  $G\mathen{\tt Fib}_X$, we have $f\circ\beta(x)\in Y_{g}^\prime$ for all $x\in U$. Then
\begin{align}
\pi^\prime\circ \mathcal{Y}_{g^{-1}}^\prime\circ f \circ \beta=\pi^\prime \circ f \circ \mathcal{Y}_{g^{-1}}\circ \beta=\pi\circ \mathcal{Y}_{g^{-1}}\circ \beta.
\end{align}
Thus we  have that $\pi^\prime\circ \mathcal{Y}_{g^{-1}}^\prime\circ f \circ \beta$ is an open embedding and have $h_{g^{-1},f\circ\beta}=h_{g^{-1}, \beta}$. This means that $\Gamma f(\Gamma Y_g)\subseteq \Gamma Y_g^\prime$. We now need to check that $\Gamma f$ obeys the commuting diagram \ref{eq:Commutating_Partial_Morphism}. So letting, $[\beta]_x\in \Gamma Y_{g^{-1}}$, (for convenience we have swapped $g$ and $g^{-1}$), we have
\begin{align}
\Gamma f \circ \Gamma \mathcal{Y}_g[\beta]_x=[f\circ \mathcal{Y}_g\circ \beta\circ h_{g,\beta}^{-1}]=[\mathcal{Y}_g^\prime\circ f \circ \beta\circ h_{g,f\circ \beta}^{-1}]=\Gamma \mathcal{Y}_g^\prime \circ \Gamma f[\beta]_x.
\end{align}
This shows that $\Gamma f$ is indeed a morphism in $G^{d}\mathen{\tt Eta}_X$. As mentioned, the analogous arguments apply for $J^r f$.
\subsection*{Proof of Lemma~\ref{th:open_embeddings_preserved}} \label{ap:Functors_preserve_embeddings_PA}
We want to show that $\Gamma$ and $J^r$ preserve embeddings of partial actions, in accordance with Lemma~\ref{th:open_embeddings_preserved}. Let us do this for $\Gamma$, noting that the proof for $J^r$ is analogous. Let $\iota:Q\rightarrow Y$ be an embedding of partial actions, meaning that it is an embedding and that $\mathcal{Q}_g=\iota^{-1}(Y_g)$. We want to show that $\Gamma Q_g=(\Gamma \iota)^{-1}(\Gamma Y_g)$. Let $[\beta]_x\in  \Gamma \mathcal{Y}_g$, such that $\Gamma \iota_Q[\tilde \beta]_x=[\beta]_x$ for some $[\tilde \beta]_x \in \Gamma Q$. We, first, want to show that $[\tilde \beta]_x\in \Gamma Q_g$. Since $\Gamma \iota_Q$ is an injection, $[\tilde \beta]_x$ is the unique element mapping into $[\beta]_x$. Explicitly we let $\beta=\iota\circ \tilde \beta$, meaning that  $\iota\circ \tilde \beta(x)\in Y_g$ for all $x\in U$. Thus $\tilde \beta(x)\in Q_g$, since $\iota$ is an embedding of partial actions. Finally, we use that $h_{g^{-1},\tilde \beta}=h_{g^{-1},\iota\circ \tilde\beta}$ to show that $h_{g^{-1},\tilde \beta}$ must be an open embedding. From this we can deduce that $[\tilde \beta]_x\in \Gamma Q_g$, and hence that $\mathcal{Q}_g=\iota^{-1}(Y_g)$.

\subsubsection*{Proof of Proposition~\ref{th:natural_transformations_pa}} \label{ap:natural_transformations_pa} 
The only non-trivial thing to check here is that
the claimed morphisms $j^r$ are indeed morphisms.

For $j^r:\Gamma Y\rightarrow \Gamma J^r Y:[\alpha]_x\mapsto [j^r\alpha]_x$, we, firstly, need to show that if $[\alpha]_x\in \Gamma Y_g$ then $[j^r\alpha]_x\in \Gamma J^r Y_g$. Assuming then that $[\alpha]_x\in \Gamma Y_g$, we have $j^r\alpha(x)=j^r_x(\alpha)\in J^rY_g$, by the similarities in the definitions of $J^r Y_g$ and $\Gamma Y_g$. We then have 
\begin{multline}
\pi^r\circ J^r\mathcal{Y}_{g^{-1}}\circ j^r\alpha=\pi^r\circ j^r(\mathcal{Y}_{g^{-1}}\circ \alpha\circ h_{g^{-1},\alpha}^{-1})\circ h_{g^{-1},\alpha}=\pi\circ \mathcal{Y}_{g^{-1}}\circ \alpha\circ h_{g^{-1},\alpha}^{-1}\circ h_{g^{-1},\alpha}\\=\pi\circ \mathcal{Y}_{g^{-1}}\circ \alpha
\end{multline}
Since $\pi\circ \mathcal{Y}_{g^{-1}}\circ \alpha$  is an open embedding, so is $\pi^r\circ J^r\mathcal{Y}_{g^{-1}}\circ j^r\alpha$, and thus $[j^r\alpha]_x\in \Gamma J^r Y_g$. We also have that $h_{g,j^r\alpha}=h_{g,\alpha}$. We now need to show that $j^r$ is such that the Diagram \ref{eq:Commutating_Partial_Morphism} commutes. Let $[\alpha]_x\in  \Gamma Y_g$, then
\begin{multline}
\Gamma J^r\mathcal{Y}_g\circ j^r[\alpha]_x=[J^r\mathcal{Y}_g\circ j^r\alpha\circ h_{g,j^r\alpha}^{-1}]_{h_{g,j^r\alpha}(x)}=[j^r(\mathcal{Y}_{g^{-1}}\circ \alpha\circ h_{g,\alpha}^{-1})\circ h_{g,\alpha}\circ h_{g,j^r\alpha}^{-1}]_{h_{g,j^r\alpha}(x)}\\=[j^r(\mathcal{Y}_{g^{-1}}\circ \alpha\circ h_{g,\alpha}^{-1})]_{h_{g,\alpha}(x)}=j^r\circ \Gamma \mathcal{Y}_{g^{-1}}[\alpha]_x
\end{multline}
Thus, $j^r$ is indeed a morphism in $G^{d}\mathen{\tt Eta}_X$.


\section{Proofs for \S\ref{sec:Constraints} in conjunction with \S\ref{sec:partial_actions}} \label{app:Proof_S3_5}
We give proofs in the most general case of cormeronomic constraints with a group action present.  
\subsection*{Proof of Theorem~\ref{th:Existence_of_sheaf_pullback}} \label{ap:Existence_of_sheaf_pullback}
We want to show that the pullback of $\Gamma \iota_Q:\Gamma Q\rightarrow \Gamma J^r Y$ and $j^r:\Gamma Y \rightarrow \Gamma J^r Y$ exists in $G^{d}\mathen{\tt Eta}$. We define the cone $((E^Q, p^Q, \mathcal{E}^Q),\{P^Q_Y:E^Q\rightarrow \Gamma Y, P^Q_Q: E^Q \rightarrow\Gamma Q\})$ and show it is the limit of this pullback. 

We let $E^Q$ be the topological space defined by the pullback in ${\tt Top}$, which exists. As a set
\begin{align} \label{eq:Ecirc_as_set}
E^Q=\{([\alpha]_x,[\beta]_x)\in \Gamma Y \times \Gamma Q\mid j^r[\alpha]_x=\Gamma \iota_Q[\beta]_x\}.
\end{align}
and $P^Q_Y:([\alpha]_x,[\beta]_x)\mapsto [\alpha]_x$ and $P^Q_Q:([\alpha]_x,[\beta]_x)\mapsto [\beta]_x$. We let $E^Q_g=(P^Q_Y)^{-1}(\Gamma Y_g)$ and we let $\mathcal{E}_g^Q$ be the unique maps such that $P^Q_Y\circ \mathcal{E}_g^Q=\Gamma Y_g\circ P^Q_Y$, which combined form a valid partial action, $\mathcal{E}^Q=(\{E^Q_g\}_{g\in G^d},\{\mathcal{E}^Q_g\}_{g\in G^d})$. Explicitly $\mathcal{E}_g:([\alpha]_x,[\beta]_x)\mapsto (\Gamma \mathcal{Y}_g [\alpha]_x, \Gamma \mathcal{Q}_g[\beta]_x)$, which works since $\iota_Q$ is an embedding of partial actions.

Now let us show that it is indeed the limit of the pullback.  Let $((E^\prime,p^\prime,\mathcal{E}^\prime),\{P^\prime_Y:E^\prime\rightarrow \Gamma Y, P^\prime_Q:E^\prime\rightarrow \Gamma Q\})$ be another cone. We define the map of sets $u:E^\prime\rightarrow E^Q: e^\prime\mapsto (P^\prime_Y(e^\prime),P^\prime_Q(e^\prime))$. We then have that $P_Y^Q\circ u=P_Y^\prime$ and $P_Q^Q \circ u=P_Q^\prime$. Since $P_Y^Q$ is an embedding, from the first of these equations we get that $u$ is continuous and it is unique. It can also, trivially, be used to show that $u$ is an \'etal\'e morphism. Finally, to make sure it is actually in $G^{d}\mathen{\tt Eta}_X$, we need to ensure it interacts correctly with the partial actions. Since $P_Y^Q\circ u(E^\prime_g)=P_Y^\prime(E^\prime_g)\subseteq \Gamma Y_g$, we have that $u(E^\prime_g)\subseteq E^Q_g$. The explicit form of $\mathcal{E}_g^Q$, and the fact that $P^\prime_Y$ and $P_Q^\prime$ are morphisms in $G^{d}\mathen{\tt Eta}_X$, tells us that so too is $u$. 

\subsection*{Proof of Propositions~\ref{th:limit_coholonomic_no_group_action} and~\ref{th:limit_comeronomic_no_group_action}} \label{app:Proof_Limits_Coho_Comer}
We now turn our attention to proving Proposition~\ref{th:limit_comeronomic_no_group_action} and consequently Proposition~\ref{th:limit_coholonomic_no_group_action}. We split this proof into a series of Lemmas.
\begin{mylem*}\label{lem:limit_of_pullback_comero} The limit of the pullback diagram
\begin{equation} \label{eq:pullbac_defining_S}
\begin{tikzcd}
 & R\arrow{d}{\iota_R}\\
J^rY\arrow{r}{J^r\tau_Z} &J^rZ
\end{tikzcd}
\end{equation}
exists in $G\mathen{\tt Fib}_X$; denoting it by the cone $((S,\psi,\mathcal{S}),\{\iota_S:S\rightarrow J^rY,\kappa:S\rightarrow R\})$, then $\iota_S$ is an embedding and $\kappa$ is a surjective submersion.
\end{mylem*}
\begin{proof}
In {\tt Man} the limit of this diagram exists, the $\iota_S$ defined by this pullback is an embedding, and the $\kappa$ is a surjective submersion. We take $S$, $\iota_S$ and $\kappa$ as defined by this pullback in ${\tt Man}$.
As a set, we have that 
\begin{align}\label{eq:S_as_set1}
S=\{(j^r_x\alpha,r)\in J^rY\times R| J^r\tau_Z(j^r_x\alpha)=\iota_R(r)\}
\end{align} 
with $\iota_s:(j^r_x\alpha,r)\mapsto j^r_x\alpha$, and $\kappa:(j^r_x\alpha,r)\mapsto r$, which are both fibred morphisms. We define $\mathcal{S}$ in the same way in which we defined $\mathcal{E}^Q$ above. That is, let $S_g=\iota_S^{-1}(J^rY_g)$. Let $\mathcal{S}_g$ the unique map such that $\iota_S\circ \mathcal{S}_g=J^rY_g\circ \mathcal{S}_g$. Explicitly, $\mathcal{S}_g:(j^r_x\alpha,r)\mapsto (\Gamma\mathcal{Y}_g j^r_x\alpha, \mathcal{R}_g r)$, which is valid since $\iota_R$ is an embedding of partial actions and $J^r\tau_Z$ is a morphism of partial actions. This makes $\iota_S$ an embedding of partial actions, and $\kappa$ a morphism of partial actions.

The fact that this construction indeed leads to a limit, follows from the same arguments as for $(E^Q, p^Q, \mathcal{E}^Q)$ above.
\end{proof}

\begin{mylem*}
The limit of the equaliser diagram
\begin{equation} \label{eq:S_equaliser}
\begin{tikzcd}[column sep=3.5em]
S \arrow[shift right,swap]{r}{ \Omega\circ \kappa}\arrow[shift left]{r}{\pi^{r,0}\circ \iota_S} & Y
\end{tikzcd}
\end{equation}
exists in $G\mathen{\tt Fib}_X$; denoting it by the cone $((Q,\nu,\mathcal{Q}),\{\iota_S^Q:Q\rightarrow S\})$, then we have that $\iota_S^Q$ is an embedding of partial actions.
\end{mylem*}
\begin{proof}
Let $\langle \pi^{r,0}\circ \iota_S,\Omega \circ \kappa \rangle:S\rightarrow Y\times Y: s\mapsto ( \pi^{r,0}\circ \iota_S(s),\Omega\circ \kappa (s))$. The map $\pi^{r,0}\circ \iota_S$ is a submersion, meaning  $\langle \pi^{r,0}\circ \iota_S,\Omega \circ \kappa \rangle$ is transverse to the diagonal map $\Delta_Y:Y\rightarrow Y\times Y$. Thus the inverse image $Q=\langle \pi^{r,0}\circ \iota_S,\Omega \circ \kappa \rangle^{-1}(\Delta_Y(Y))$ exists, with a corresponding embedding $\iota^S_Q:Q\rightarrow S$ of $Q$ into $S$.

\emergencystretch 3em
Let $\nu:=\pi^r\circ \iota_Q$, let $y\in J^rY$, and let $U_{y}\in J^rY$ be a neighbourhood of $y$, with coordinates $(x^\mu,z^a,y^i,z^a_I,y^i_I)$. In these coordinates, $Q$ is described by $(x^\mu,z^a,f^i(x^\mu,z^a,z^a_I),z^a_I,y^i_I)$, for some smooth $f^i$. From this, we see that $\nu:(x^\mu,z^a,f^i(x^\mu,z^a,z^a_I),z^a_I,y^i_I)\mapsto (x^\mu)$ is a surjective submersion.

As a set, we have that 
\begin{align} \label{eq:Q_as_set1}
Q=\{s\in S\mid \pi^{r,0}\circ \iota_S(s)=\Omega\circ \kappa(s)\}.
\end{align}
To define $\mathcal{Q}$ we first define $Q_g=(\iota_Q^S)^{-1}(S_g)$. Then, as before, we let $\mathcal{Q}_g$ be the unique map such that $\iota_Q^S\circ \mathcal{Q}_g=\mathcal{S}_g\circ \iota_Q^S$. The fact that such $\mathcal{Q}_g$ exist  can be seen from the form of $Q$ and the fact that $\pi^{r,0}$, $\iota_S$, $\Omega$, and $\kappa$ are all morphisms in $G\mathen{\tt Fib}_X$. This then makes $\iota_S^Q$ an embedding of partial actions. 

The universality property then follows that of $(E^Q, p^Q,\mathcal{E}^Q)$.
\end{proof}
Let $\iota_Q:=\iota_S\circ \iota_S^Q$; since both maps in the composition are embeddings of partial actions, so is $\iota_Q$. Let $f^R_Q:=\kappa\circ \iota_S^Q$, we then have 

\begin{mylem*}
The triple $((Q,\nu,X,\mathcal{Q}),\{\iota_Q,f^R_Q\})$ forms the limit of Diagram \ref{eq:Commutative_Comeronomic}
\end{mylem*} 
\begin{proof}
First let us show that $((Q,\nu,X,\mathcal{Q}),\{\iota_Q,f^R_Q\})$ is indeed a cone of Diagram~\ref{eq:Commutative_Comeronomic}.  For this to hold we need $\pi^{r,0}\circ \iota_Q=\Omega \circ f^R_Q$, which follows from the equaliser in Diagram~\ref{eq:S_equaliser}, and $J^r\tau_Z \circ \iota_Q=\iota_R\circ f^r_Q$ which follows from the pullback~\ref{eq:pullbac_defining_S}.

To show that this cone is a limit, suppose we have another cone $((Q^\prime,\nu^\prime,X,\mathcal{Q}),\{f_{\hat \pi^r}, f_R\})$. From Eqs.~\ref{eq:S_as_set1} and~\ref{eq:Q_as_set1}, we can write $Q$ as 
\begin{align}
Q=\{(j^r_x\alpha,r)\in J^r\hat \pi\times R\mid J^r\tau_Z(j^r_x\alpha)=\iota_r(r), \Omega(r)=\pi^{r,0}(j^r_x\alpha)\}.
\end{align}
We let $u:Q^\prime\rightarrow Q:q^\prime \mapsto (f_{\hat \pi^r}(q^\prime),f_R(q^\prime))$; the standard argument shows that this is a mediating morphism.
\end{proof}

\subsection*{Proof of Theorems~\ref{th:isomorphism_of_sheaves_coholonomic} and \ref{th:isomorphism_of_sheaves_comeronomic}} \label{ap:Isomorphism_unconstrained_constrained}
We want to prove Theorem~\ref{th:isomorphism_of_sheaves_comeronomic} and consequently Theorem~\ref{th:isomorphism_of_sheaves_coholonomic}. Namely, we want to show that there exists an isomorphism between the \'etal\'e spaces $(E^Q, p^Q, \mathcal{E}^Q)$ and $( E^R,  p^R, \mathcal{E}^R)$. This follows from the structure of a series of cones and limits. We start by noting that $((\Gamma Q,\Gamma\nu,\mathcal{Q}),\{\Gamma f^r_Q, \Gamma \iota_Q\})$ is the limit of Diagram~\ref{eq:Commutative_Comeronomic} in $G^{d}\mathen{\tt Eta}$, something which can be shown explicitly following the standard arguments used previously. But, $(( E^R, \tilde p^R,\mathcal{E}^R),\{ P^R_R,j^r\circ\Gamma \Omega\circ  P^R_R\})$ is manifestly a cone of this diagram since, for instance,
\begin{multline}
\Gamma J^r\tau_Z\circ j^r\circ\Gamma \Omega\circ  P^R_R=j^r\circ \Gamma (\tau_Z\circ \Omega)=j^r\circ\Gamma \zeta^{r,0}\circ \Gamma\iota_R\circ  P^R_R\\=j^r\circ P^R_Z=\Gamma\iota_R\circ  P^R_R.
\end{multline}
 We denote the corresponding mediating morphism, $\mathcal{N}:\Gamma R\rightarrow \Gamma Q$. Since, 
 \begin{align}
\Gamma \iota_Q\circ\mathcal{N}=j^r\circ \Gamma \Omega\circ {P}^R_R,
\end{align}
 $(( E^R,  p^R, \mathcal{E}^R),\{\mathcal{N},\Gamma \Omega\circ  P^R_R\})$ is a cone of the diagram defining $E^Q$.
This means we have a mediating morphism $\mathcal{I}: E^R\rightarrow E^Q$. In a similar vein, $(( E^Q,  p^Q, \mathcal{E}^Q),\{\Gamma f^R_Q\circ  P^Q_Q,\Gamma \tau_Z\circ  P^Q_Y\})$ is a cone of the diagram defining $ E^R$ since
\begin{align}
\Gamma \iota_R\circ \Gamma f^R_Q\circ P^Q_Q=\Gamma J^r \tau_Z \circ \Gamma\iota_Q \circ P^Q_Q=\Gamma J^r \tau_Z \circ j^r\circ P^Q_Y=j^r\circ \Gamma\tau_Z \circ P^Q_Y
\end{align}
 and thus we have a mediating morphism $\tilde{\mathcal{I}}: E^Q\rightarrow\tilde E^R$. With this we can show that $\tilde{\mathcal{I}}\circ \mathcal{I}$ is the identity, since
\begin{align}
 P^R_Z\circ \tilde{\mathcal{I}}\circ \mathcal{I}=\Gamma \tau_Z\circ P^Q_Y\circ \mathcal{I}=\Gamma \tau_Z \circ \Gamma \Omega \circ  P_R^R=\Gamma\zeta^{r,0}\circ \Gamma \iota_R\circ  P_R^R= P_Z^R ,
\end{align} 
which, since $ P^R_Z$ is an embedding, shows that $\tilde{\mathcal{I}}\circ \mathcal{I}$ is the identity. The statement that $\mathcal{I}\circ\tilde{\mathcal{I}}$ is the identity holds in a similar vein, since
\begin{align}
P^Q_Y\circ \mathcal{I}\circ \tilde{\mathcal{I}}=\Gamma \Omega \circ  P^R_R\circ \tilde{\mathcal{I}}=\Gamma \Omega\circ \Gamma f_Q^R\circ P^Q_Q=\Gamma \pi^{r,0}\circ \Gamma \iota_Q\circ P^Q_Q=P^Q_Y.
\end{align}
Thus $\mathcal{I}$ and $\tilde{\mathcal{I}}$ are mutually inverse and form isomorphisms. 
\bibliography{./Draft_references}

\providecommand{\href}[2]{#2}\begingroup\raggedright\begin{thebibliography}{10}

\bibitem{Ivanov_Ogievetsky_1975}
E.~A. Ivanov and V.~I. Ogievetsky, \emph{The inverse higgs phenomenon in
  nonlinear realizations}, {\emph{Theoretical and Mathematical Physics}
  {\bfseries 25} (1975) 1050}.

\bibitem{Coleman:1969sm}
S.~R. Coleman, J.~Wess and B.~Zumino, \emph{{Structure of phenomenological
  Lagrangians. 1.}},
  \href{https://doi.org/10.1103/PhysRev.177.2239}{\emph{Phys. Rev.} {\bfseries
  177} (1969) 2239}.

\bibitem{Callan:1969sn}
J.~Callan, Curtis~G., S.~R. Coleman, J.~Wess and B.~Zumino, \emph{{Structure of
  phenomenological Lagrangians. 2.}},
  \href{https://doi.org/10.1103/PhysRev.177.2247}{\emph{Phys. Rev.} {\bfseries
  177} (1969) 2247}.

\bibitem{lee_2009}
J.~M. Lee, \emph{Manifolds and Differential Geometry}. American Mathematical
  Society, 1~ed., 2009,
  \href{https://doi.org/http://dx.doi.org/10.1090/gsm/107}{http://dx.doi.org/10.1090/gsm/107}.

\bibitem{Kolar_1993}
I.~Kolar, P.~W. Michor and J.~Slovak, \emph{Natural Operations in Differential
  Geometry}. Springer-Verlag, 1~ed., 1993,
  \href{https://doi.org/10.1007/978-3-662-02950-3}{10.1007/978-3-662-02950-3}.

\bibitem{Saunders_1989}
D.~J. Saunders, \emph{The Geometry of Jet Bundles}. Cambridge University Press,
  1~ed., Mar, 1989,
  \href{https://doi.org/10.1017/CBO9780511526411}{10.1017/CBO9780511526411}.

\bibitem{Sardanashvily:2009br}
G.~Sardanashvily, \emph{{Fibre Bundles, Jet Manifolds and Lagrangian Theory:
  Lectures for Theoreticians}},
  \href{https://arxiv.org/abs/0908.1886}{{\ttfamily 0908.1886}}.

\bibitem{olver2000applications}
P.~J. Olver, \emph{Applications of Lie groups to differential equations},
  vol.~107. Springer Science \& Business Media, 2000,
  \href{https://doi.org/10.1007/978-1-4684-0274-2}{10.1007/978-1-4684-0274-2}.

\bibitem{Wedhorn2016}
T.~Wedhorn, \emph{Manifolds, Sheaves, and Cohomology}. Springer Spektrum, 2016,
  \href{https://doi.org/10.1007/978-3-658-10633-1}{10.1007/978-3-658-10633-1}.

\bibitem{Krupkova09}
O.~Krupkov{\'{a}}, \emph{The nonholonomic variational principle},
  \href{https://doi.org/10.1088/1751-8113/42/18/185201}{\emph{Journal of
  Physics A: Mathematical and Theoretical} {\bfseries 42} (2009) 185201}.

\bibitem{Krupkova00}
O.~Krupková, \emph{Higher-order mechanical systems with constraints},
  \href{https://doi.org/10.1063/1.533411}{\emph{Journal of Mathematical
  Physics} {\bfseries 41} (2000) 5304}.

\bibitem{Krupkova97}
O.~Krupková, \emph{Mechanical systems with nonholonomic constraints},
  \href{https://doi.org/10.1063/1.532196}{\emph{Journal of Mathematical
  Physics} {\bfseries 38} (1997) 5098}.

\bibitem{CapSlovak_2009}
A.~\v{C}ap and J.~Slov\'{a}k, \emph{Parabolic Geometries I: Background and
  General Theory}. American Mathematical Society, 1~ed., 2009,
  \href{https://doi.org/10.1090/surv/154}{10.1090/surv/154}.

\bibitem{Nicolis:2015sra}
A.~Nicolis, R.~Penco, F.~Piazza and R.~Rattazzi, \emph{{Zoology of condensed
  matter: Framids, ordinary stuff, extra-ordinary stuff}},
  \href{https://doi.org/10.1007/JHEP06(2015)155}{\emph{JHEP} {\bfseries 06}
  (2015) 155} [\href{https://arxiv.org/abs/1501.03845}{{\ttfamily
  1501.03845}}].

\bibitem{Goon:2012dy}
G.~Goon, K.~Hinterbichler, A.~Joyce and M.~Trodden, \emph{{Galileons as
  Wess-Zumino Terms}},
  \href{https://doi.org/10.1007/JHEP06(2012)004}{\emph{JHEP} {\bfseries 06}
  (2012) 004} [\href{https://arxiv.org/abs/1203.3191}{{\ttfamily 1203.3191}}].

\bibitem{Nicolis:2013lma}
A.~Nicolis, R.~Penco and R.~A. Rosen, \emph{{Relativistic Fluids, Superfluids,
  Solids and Supersolids from a Coset Construction}},
  \href{https://doi.org/10.1103/PhysRevD.89.045002}{\emph{Phys. Rev. D}
  {\bfseries 89} (2014) 045002}
  [\href{https://arxiv.org/abs/1307.0517}{{\ttfamily 1307.0517}}].

\bibitem{ABADIE200314}
F.~Abadie, \emph{Enveloping actions and takai duality for partial actions},
  \href{https://doi.org/https://doi.org/10.1016/S0022-1236(02)00032-0}{\emph{Journal
  of Functional Analysis} {\bfseries 197} (2003) 14}.

\bibitem{10.2307/24718783}
J.~Quigg and I.~Raeburn, \emph{Characterisations of crossed products by partial
  actions}, {\emph{Journal of Operator Theory} {\bfseries 37} (1997) 311}
  [\href{https://arxiv.org/abs//funct-an/9604001}{{\ttfamily
  /funct-an/9604001}}].

\bibitem{Low:2001bw}
I.~Low and A.~V. Manohar, \emph{{Spontaneously broken space-time symmetries and
  Goldstone's theorem}},
  \href{https://doi.org/10.1103/PhysRevLett.88.101602}{\emph{Phys. Rev. Lett.}
  {\bfseries 88} (2002) 101602}
  [\href{https://arxiv.org/abs/hep-th/0110285}{{\ttfamily hep-th/0110285}}].

\end{thebibliography}\endgroup
\end{document}